\documentclass{article}
\usepackage[utf8]{inputenc}
\usepackage[T1]{fontenc} 
\usepackage{lmodern} 
\usepackage[english]{babel} 
\usepackage{csquotes} 
\usepackage{amsmath,amsfonts,amssymb} 
\usepackage{numprint}
\usepackage[shortlabels]{enumitem}
\usepackage{tikz}
\usepackage{graphicx}
\usetikzlibrary{arrows, positioning, fit, shapes.geometric}
\usepackage{pgfplots}
\pgfplotsset{compat=newest}
\usepgfplotslibrary{groupplots}
\usepackage{subcaption}

\usepackage{listings} 
\usepackage[backend=biber,citestyle=numeric,bibencoding=utf8,giveninits=true,url=false,maxbibnames=5,autolang=other]{biblatex}
\newcommand{\doctitle}{Mathematical Analysis and Algorithms for Federated Byzantine Agreement Systems}
\newcommand{\docauthors}{Andr{\'e} Gaul\footnotemark[1], 
Ismail Khoffi\footnotemark[1], 
J{\"o}rg Liesen\footnotemark[1],
and Torsten St{\"u}ber\footnotemark[2]}

\usepackage{hyperref} 
\hypersetup{
  draft=false,
  pdftitle={\doctitle},
  pdfauthor={André Gaul, Ismail Khoffi, Jörg Liesen, Torsten Stüber},
  pdflang={en},
  colorlinks=true,
  linkcolor=black,
  citecolor=black,
  filecolor=black,
  urlcolor=black
}

\newlength\figurewidth
\newlength\figureheight

\allowdisplaybreaks[4]

\newcommand{\N}{\mathbb N}

\usepackage{amsmath,amsfonts,amssymb} 
\usepackage{amsthm} 

\newtheorem{thm}{Theorem}[section]
\newtheorem{lem}[thm]{Lemma}

\newtheorem{cor}[thm]{Corollary}

\theoremstyle{definition}
\newtheorem{ex}[thm]{Example}
\newtheorem{rmk}[thm]{Remark}
\newtheorem{definition}[thm]{Definition}

\title{\doctitle}
\author{\docauthors}

\tikzset{
  treenode/.style = {align=center, inner sep=0pt},
  black_n/.style = {treenode, circle, white, draw=black,
    fill=black, text width=0.6em},
  white_n/.style = {treenode, circle, black, draw=black, 
    text width=0.6em, thick},
  dashed_n/.style = {treenode, circle, black, draw=black, 
    text width=0.6em, densely dashed},
  dashed_e/.style={edge from parent/.style={densely dashed,draw}}
}
\definecolor{plot_color_0}{rgb}{0.886274509803922,0.290196078431373,0.2}
\definecolor{plot_color_1}{rgb}{0.203921568627451,0.541176470588235,0.741176470588235}

\providecommand{\keywords}[1]
{{\small	
\textbf{\textbf{Keywords --}} #1
}}

\begin{document}

\maketitle

\renewcommand{\thefootnote}{\fnsymbol{footnote}}
\footnotetext[1]{TU Berlin, Institute of Mathematics, Stra{\ss}e des 17. Juni 136, 10623 Berlin, Germany, {\tt andre@gaul.io}, 
{\tt ismail.khoffi@gmail.com}, {\tt liesen@math.tu-berlin.de}} 
\footnotetext[2]{SatoshiPay GmbH, M{\"u}hlenstr. 8a, 14167 Berlin, Germany, {\tt torsten@satoshipay.io}}

\renewcommand{\thefootnote}{\arabic{footnote}}

\begin{abstract}
We give an introduction to \emph{federated Byzantine agreement systems (FBAS)} with many examples ranging from small ``academic'' cases to the current Stellar network. We then analyze the main concepts from a mathematical and an algorithmic point of view. Based on work of Lachowski~\cite{Lac19} we derive algorithms for \emph{quorum enumeration}, checking \emph{quorum intersection}, and computing the \emph{intact nodes} with respect to a given set of ill-behaved (Byzantine) nodes. We also show that from the viewpoint of the \emph{intactness probability} of nodes, which we introduce in this paper, a hierarchical setup of nodes 
is inferior to an arrangement that we call a \emph{symmetric simple FBAS}. All algorithms described in this paper are implemented in the Python package \emph{Stellar Observatory}\/\footnote{\url{https://github.com/andrenarchy/stellar-observatory}}, which is also used in some of the computed examples.
\end{abstract}


\section{Introduction}

The main idea of blockchain technology is to maintain a distributed ledger in a decentralized way. This means that the ledger is replicated among a set of different parties or \emph{nodes}, and no central entity has the authority on the content of the ledger~\cite{Nakamoto09}.
Transactions included in the ledger are recorded in blocks which point to previous blocks, thus forming a chain of blocks.
Nodes need to reach \emph{consensus} (see, e.g.,~\cite{Pease80} or~\cite[p. 659]{CouDolKin11}) on the next block in order to be included in the chain despite facing arbitrary (Byzantine) failures of some nodes such as going offline or proposing an invalid block. 
Consequently, every blockchain system requires an underlying consensus model, which ideally is based on a thorough mathematical formalization and analysis. 

In this paper we study a consensus model introduced by Mazi\`eres in 2016~\cite{Maz16}, the \emph{federated Byzantine agreement system (FBAS)}, which is used in the \emph{Stellar\/\footnote{\url{https://www.stellar.org}} Consensus Protocol (SCP)}. The main practical advantages of the SCP in comparison with other consensus mechanisms like proof-of-work or proof-of-stake with the longest chain rule, or classical Byzantine agreement are discussed in detail in~\cite[Section~2]{Maz16} and in~\cite[Section~2]{LokLosMaz19}. 

The nodes of an FBAS conceptually form a network, where the links between the nodes are established by trust decisions that are made by each node individually, and hence in a decentralized way. A subset of nodes that a given node ``trusts'' is called a quorum slice, and each node may have several such slices. A quorum of the FBAS is a set of elements which contains at least one quorum slice for each of its members, and hence may be interpreted as a subset of nodes that ``trusts itself''. In order to avoid that contradictory statements are ratified, the quorums need to be pairwise intersecting. The quorum intersection property and several other concepts in the context of the FBAS lead to challenging mathematical and computational problems which so far have not been fully addressed in the literature. Our goal in this paper is contribute to the overall understanding of FBAS by giving a thorough mathematical introduction to the main ideas and concepts, in particular the quorum intersection problem and the intactness of nodes. We will also treat these problems in detail from an algorithmic point of view. 

The paper is organized as follows: In Section~\ref{sec:FBAS} we give the FBAS definition of Mazi\`eres and illustrate it with several examples. We then introduce and analyze the \emph{trust graph} of an FBAS and derive a mathematical description of the current Stellar network. In Section~\ref{sec:alg-quorums} we show that the quorum intersection decision problem is NP-complete. A similar result was recently shown by Lachowski~\cite{Lac19}, but here we show NP-completeness even when we restrict the decision problem to the class of \emph{simple FBAS} (see Definition~\ref{def:simple}). We then discuss algorithms for quorum enumeration and checking quorum intersection. The main ideas of these algorithms also appeared in the work of Lachowski~\cite{Lac19}, and some have been implemented in \emph{Stellar Core}\/\footnote{\url{https://github.com/stellar/stellar-core}}, the ``backbone'' of the Stellar network. Based on these previous works, we give an in-depth treatment of both algorithms.

In Section~\ref{sec:intact} we analyze the intactness of nodes, characterize the intact nodes in \emph{symmetric simple FBAS}, and compare the original intactness definition of Mazi\`eres with newer definitions given in~\cite{LokLosMaz19,LosGafMaz19}. In Section~\ref{sec:algo-intact} we treat intactness algorithmically, and based on results from Section~\ref{sec:alg-quorums}, we derive an algorithm for computing the intact nodes when a set of ill-behaved (Byzantine) nodes is given. The time complexity of this algorithm is (only) twice as long as determining quorum intersection. Finally, in Section~\ref{sec:intact-prob} we look at intactness from a probabilistic point of view, and we study the probability that a node is intact in different scenarios of ill-behaved nodes.


\section{Federated Byzantine Agreement Systems}\label{sec:FBAS}

In this section we first review the basic definition and give some simple examples of federated Byzantine agreement systems, then discuss their relationship to certain graphs, and finally derive a mathematical description of the current Stellar network.   

\subsection{Basic definitions and examples}

We start with the following definition, originally given by Mazi\`eres~\cite[p.~4]{Maz16}.


\begin{definition}[FBAS and quorum]\label{def:FBAS}
A \emph{federated Byzantine agreement system (FBAS)} is a pair $(V,S)$ consisting of a finite set of nodes $V$ and a \emph{quorum function} $S:V\to {\mathcal P}({\mathcal P}(V)) \setminus \{\emptyset\}$, where for each $v\in V$ and $s\in S(v)$ we require that $v\in s$. Each set $s\in S(v)$ is called a \emph{quorum slice} of the node $v$. \\
A nonempty set of nodes $Q\subseteq V$ is called a \emph{quorum} in $(V,S)$ if for each $v\in Q$ there exists a quorum slice $s\in S(v)$ with $s\subseteq Q$.
\end{definition}

Note that we require that $S(v)\not=\emptyset$ for every node $v$. Therefore $Q=V$ is a quorum in $(V,S)$, so that each FBAS has at least one quorum. Also note that if $Q_1,Q_2$ are quorums, then $Q_1\cup Q_2$ is a quorum as well. 

\begin{rmk}
    In a more recent paper on Stellar as a global payment network, the authors define a quorum to be ``a non-empty set of nodes encompassing at least one quorum slice of each \emph{non-faulty} member''~\cite[p.~4]{LokLosMaz19}. In contrast to Definition~\ref{def:FBAS}, this definition of quorums includes the individual behavior of the nodes (which is unknown a priori and uncontrollable). In this paper we will focus on the FBAS as originally defined by Maz\`ieres, and thus consider the individual node behavior only later, when defining the concept of \emph{intactness} (see Definition~\ref{def:intact}). A comparison of the original and the new concepts is given in Section~\ref{sec:intact-compare}. We stress that in the analysis of consensus mechanisms it is very important to pay great attention to details. Already Lamport, Shostak and Pease in their classical paper on the Byzantine generals problem from 1982 ``strongly advise the reader to be very suspicious of ... nonrigorous reasoning'' and write that they ``know of no area in computer science or mathematics in which informal reasoning is more likely to lead to errors than in the study of this type of algorithm''~\cite[p.~385]{LamShoPea82}.
\end{rmk}

A quorum slice $s\in S(v)$ is interpreted as a subset of the nodes that the node~$v$ ``trusts''. The set $S(v)$ should be chosen so that $v$ will agree with any statement that is unanimously agreed upon by any of the slices. A quorum $Q$ of the FBAS again is a subset of the nodes, where for each element $v\in Q$ at least one of its quorum slices (i.e., at least one subset of ``trusted'' nodes) is completely contained in $Q$ as well. Thus, a quorum is a set of nodes that ``trusts itself''. In practice quorums are not known a priori but discovered by the nodes when exchanging messages. Each node sends its quorum slices with every message. A receiving node considers a message to satisfy a quorum when the set of the sending nodes contains a quorum slice of each of its elements.

In a nutshell (see~\cite[Section~5.1]{Maz16} for details), a quorum $Q$ \emph{ratifies} a statement if every node $v\in Q$ asserts that this statement is true. \emph{Consensus} means that two contradicting statements should not be ratified at the same time. Thus, the quorums of the FBAS should pairwise intersect, which naturally leads to the following concept introduced by Mazi\`eres~\cite[p.~8]{Maz16}.

\begin{definition}[quorum intersection]
    An FBAS $(V,S)$ has \emph{quorum intersection} if any two of its quorums have a nonempty intersection.
\end{definition}

Note that quorum intersection is not guaranteed without further assumptions. Moreover, the decision problem whether quorum intersection holds is NP-complete; see Section~\ref{sec:NPcomplete} below for details. As described in~\cite[Section~6.2.1]{LokLosMaz19}, 
within \emph{Stellar Core} the quorum intersection property of the underlying FBAS is regularly checked with an algorithm based on an idea of Lachowski~\cite{Lac19}, which we will analyze in detail in Sections~\ref{sec:quorum-enumeration} and~\ref{sec:quorum-intersection}. 

\begin{rmk}
    If $Q_1,\dots,Q_m$ are the quorums of an FBAS $(V,S)$ with quorum intersection, then $(V, \{Q_1,\dots,Q_m\})$ is a (classical) \emph{quorum system}. Such systems have been studied since the late 1970s in the context of distributed computing; see, e.g., the monograph~\cite{Vuk12} for an overview. The quorum system $(V,\{Q_1,\dots,Q_m\})$ therefore may be analyzed using the concepts introduced in this area, including \emph{load}, \emph{capacity}, and \emph{availability}; see, e.g.,~\cite{MalRei97,MalReiWoo00,NaoWoo98}. A detailed analysis of the relationship between FBAS and classical Byzantine quorum systems was recently given in~\cite{GarGot18}.
\end{rmk}

Let us give a few examples of federated Byzantine agreement systems.

\begin{ex}\label{ex:trivial}
Let $V=\{1,\dots,n\}$ for some $n\geq 2$.\\
(1) Let $S(i)=\{\{i\}\}$ for $i=1,\dots,n$. Then every nonempty subset $Q\subseteq V$ is a quorum, since for every node $i$ its only quorum slice is given by $\{i\}$, and we trivially have $\{i\}\subseteq Q$. Clearly, $(V,S)$ does not have quorum intersection.\\
(2) Let $S(i)=\{V\}$ for $i=1,\dots,n$. Now $Q=V$ is the only quorum and $(V,S)$ has quorum intersection.
\end{ex}    
    
\begin{ex}\label{ex:two}
    Consider $V=\{1,\dots,n\}$ for some $n\geq 2$ and $S(i)=\{\{i,j\}\mid j=1,\dots,n,\,j\neq i\}$ for $i=1,\dots,n$. Let $U\subseteq V$ be any subset containing at least two nodes. If we take any two nodes from $U$, say $\ell$ and $m$ with $\ell\neq m$, then $\{\ell,m\}\in S(\ell)$ and $\{\ell,m\}\in S(m)$, which shows that $U$ is a quorum. If $n\geq 4$, then $(V,S)$ does not have quorum intersection, since, e.g., $\{1,2\}$ and $\{3,4\}$ are quorums.
\end{ex}

\begin{ex}[Byzantine agreement]\label{ex:Byzantine}
    This example shows that classical Byzantine agreement can be modeled as a special FBAS.
    
    Let $V=\{1,\dots,n\}$ for some $n=3m+1$, where $m\in\N$, and suppose that for $i=1,\dots,n$, each quorum slice $s\in S(i)$ contains at least $2m+1$ nodes (including~$i$). Then any quorum of the FBAS $(V,S)$ must consist of at least $2m+1$ nodes, and since there are exactly $3m+1$ nodes, any two quorums must intersect in at least $m+1$ nodes. If $m$ nodes in the intersection become Byzantine, i.e., behave arbitrarily, then there still is at least one non-Byzantine node in each intersection, which guarantees that the quorums cannot ratify contradictory statements. And, as shown in a classical paper of Lamport, Shostak, and Pease, ``no solution with fewer than $3m+1$ generals can cope with $m$ traitors''~\cite[p.~387]{LamShoPea82}; also cf. Theorem~\ref{thm:SymmetricByzantine} below.
\end{ex}

\begin{ex} (\cite[Fig.~7]{Maz16}) \label{ex:FBAS_fig7}
    Consider $V=\{1,\dots,7\}$ and $S$ with 
    \begin{align*}
        S(i)&=\{\{1,2,3,7\}\},\; i=1,2,3,\qquad S(i)=\{\{4,5,6,7\}\},\;i=4,5,6\;,
    \end{align*}
    and $S(7)=\{\{7\}\}$. This FBAS $(V,S)$ has exactly four quorums,
    \begin{align*}
        \{1,2,3,7\},\quad\{4,5,6,7\},\quad\{7\},\quad V,
    \end{align*}
    and it has quorum intersection. Note that node $7$ is contained in every quorum, and that the minimal intersection of two quorums consists only of this node.
\end{ex}

\begin{ex}
     Suppose that the FBAS is built by 4 organizations $A,B,C,D$, which have 3 nodes each, i.e., $V=A\cup B\cup C\cup D$ with
    \begin{align*}
        A&=\{a_1,a_2,a_3\},&
        B&=\{b_1,b_2,b_3\},&
        C&=\{c_1,c_2,c_3\},&
        D&=\{d_1,d_2,d_3\}.
    \end{align*}
    Suppose that the quorum slices of each node $x_i$, $i=1,2,3$, are all sets consisting of the three nodes of its own organization $X$, and exactly one node from another organization, i.e., 
    $$S(x_i)=\{X\cup \{y\}\mid y\in V\setminus X\}.$$ 
    Then each node has 9 quorum slices containing 4 nodes each, and $(V,S)$ does not have quorum intersection, since $A\cup B$ and $C\cup D$ are quorums.
\end{ex}

\begin{ex}[hierarchical quorum slices]\label{ex:hierarchy}
    Suppose that the FBAS is built by 6 organizations, $V=A\cup B\cup C\cup D\cup E\cup F$ with
    \begin{align*}
        A&=\{a_1,a_2,a_3\},\quad
        B=\{b_1,b_2,b_3\},\quad
        C=\{c_1,c_2,c_3\},\\
        D&=\{d_1,d_2,d_3\},\quad
        E=\{e_1,e_2,e_3\},\quad
        F=\{f_1,f_2,f_3,f_4,f_5\},
    \end{align*}
    and suppose that the quorum slices are built in a hierarchical way: We form \emph{all possible} subsets of $V$ by first choosing 5 of the 6 sets $A,B,C,D,E,F$, and then in each of the 5 chosen sets we choose either 2 of the 3 nodes in case of the sets $A,B,C,D,E$, or 3 of the 5 nodes in case of set $F$. Thus, each of these subsets contains either 10 or 11 nodes, depending on whether on the ``root level'' of the hierarchy the set $F$ is chosen or not. In total there are \numprint{4293} such subsets, and we denote the union of all these sets by $\underline{S}$.
    
    We define the quorum function by
    $$S(v)=\{U\in \underline{S}\mid v\in U\},\quad v\in V,$$
    then each node in the sets $A,B,C,D,E$ has \numprint{2322} quorum slices, and each node in the set $F$ has \numprint{2430} quorum slices. This results in \numprint{37888} quorums and we have quorum intersection, with the minimal size of an intersection given by 4 nodes. 
\end{ex}

As shown by the examples above, the definition of the FBAS allows quite intricate systems of quorum slices and resulting quorums. We will now introduce special FBAS classes that are of some practical relevance and can be analyzed more easily than the general case. The idea of these classes is that the quorum slices of each node are given by all possible subsets of a fixed cardinality that can be drawn from a certain fixed subset of the nodes.

\begin{definition}[simple and symmetric FBAS]\label{def:simple}
    An FBAS $(V,S)$ is called \emph{simple} when there exist functions $q: V \rightarrow {\mathcal P}(V)$ and $n: V\rightarrow \N$, such that for each node $v\in V$ we have $v\in q(v)$, $1\leq n(v)\leq |q(v)|$, and
    $$S(v)=\{U\subseteq q(v)\mid v\in U\;\mbox{and}\;|U| = n(v)\}.$$
    We denote a simple FBAS by $(V, q, n)$. A simple FBAS is called \emph{symmetric} when for all $v\in V$ we have $q(v)=V$ and $n(v)=k$ for some fixed $k$, where $1\leq k\leq |V|$, and we denote such an FBAS by $(V,k)$. 
\end{definition}

Note that Example~\ref{ex:trivial}(1) is a simple FBAS with $q(v_i)=\{v_i\}$ and $n(v_i)=1$, Example~\ref{ex:trivial}(2) is a symmetric simple FBAS with $k=n$, and Example~\ref{ex:two} is a symmetric simple FBAS with $k=2$.

\begin{lem}\label{lem:simplequorum}
    Let $(V,k)$ be a symmetric simple FBAS. Then the set of all quorums is given by $\{Q\subseteq V\,\mid\,|Q|\geq k\}$, and the FBAS has quorum intersection if and only if $k>|V|/2$.
\end{lem}
\begin{proof}
    It is easy to see that $\{Q\subseteq V\,\mid\,|Q|\geq k\}$ is the set of all quorums, since 
    for each node $v\in V$ its quorum slices are given by all subsets of $V$ containing exactly $k$ elements (including $v$ itself). If $k\leq |V|/2$, then there exist two disjoint quorums. On the other hand, if $k>|V|/2$, then the pigeonhole principle implies that the quorums are pairwise intersecting.
\end{proof}


\subsection{The trust graph}\label{sec:trust-graph}

It is tempting to think of and to visualize the network of relationships that is represented by an FBAS as a conventional directed graph, where the vertices are given by the elements of $V$, and an edge between $v_i$ and $v_j$ exists when $v_j$ is contained in a quorum slice of $v_i$. While in general there is no one-to-one correspondence between directed graphs and FBAS, the representation of an FBAS by a directed graph is a useful tool for their analysis because it describes the \emph{network of trust}. 


\begin{definition}[trust graph and SCCs]\label{def:trust-graph}
    The \emph{trust graph} of the FBAS $(V,S)$ is the directed graph $G=(V,E)$, where for every $u,v\in V$ we have $(u,v)\in E$ if 
    $v\in s$ for some $s\in S(u)$.\\
    For every $u,v\in V$ we say that $v$ is \emph{reachable} from $u$ if there is a path from $u$ to $v$ in $G$. We say that a set $U\subseteq V$ is \emph{strongly connected} in $G$ if for every $u,v\in U$ the node $u$ is reachable from $v$. Furthermore, we say that a nonempty set $C\subseteq V$ is a \emph{strongly connected component} (abbreviated \emph{SCC}) of $G$ if it is strongly connected and no proper superset of $C$ is strongly connected.
\end{definition}

Note that due to the last condition, whenever $C$ is an SCC and there is a node $v\in V$ such that every node in $C$ is reachable from $v$, and $v$ is reachable from every node in $C$, then $v\in C$. The following result follows directly from the transitivity of the reachability relation.

\begin{lem}\label{lem:SCC}
    If $C$ and $D$ are SCCs, then the following hold:
    \begin{enumerate}[(1)]
        \item $C=D$ or $C\cap D=\emptyset$.
        \item Let $c\in C$ and $d\in D$ be such that $d$ is reachable from $c$. Then every node in $D$ is reachable from every node in $C$.
    \end{enumerate}
\end{lem}

The first statement of Lemma~\ref{lem:SCC} implies that the set of SCCs form a partition of $V$ into disjoint nonempty subsets. The second statement implies that the following notion of reachability of SCCs is well-defined and is a partial order on SCCs.

\begin{definition}[maximal and greatest SCC]\label{def:SCC-props}
    Let $C$ and $D$ be SCCs. We say that $D$ is \emph{reachable} from $C$ if there are $c\in C$ and $d\in D$ such that $d$ is reachable from $c$. We say that an SCC is \emph{maximal} if no other SCC is reachable from it. We say that an SCC is the \emph{greatest} SCC if it is reachable from every other SCC.
\end{definition}

\begin{ex}[Continuation of Example~\ref{ex:FBAS_fig7}]
    The strongly connected components of the FBAS of Example~\ref{ex:FBAS_fig7} are $\{1,2,3\}$, $\{4,5,6\}$ and $\{7\}$; see Figure~\ref{fig:FBAS_fig7}, which also shows how the reachability relation carries over to the SCCs. Note that the SCC $\{7\}$ is the only maximal SCC and also the greatest SCC.
    
    \begin{figure}
        \centering
        \begin{tikzpicture}[auto,
            level distance = 1.25cm,
            node/.style={circle,fill=white,draw},
            edge_style/.style={draw=orange, line width=2, ultra thick},
            bi_dir_e/.style={edge from parent/.style={<->,>=stealth',draw}}]
            
            \node[node] (7) {7};
            \node[circle, draw=black, fit=(7), dashed] (q3) {};
            
            \node[node, below right= 6mm and 19mm of q3] (2) {2}
                child[bi_dir_e]{ node [node] (1) {1}}
                child[bi_dir_e]{ node [node] (3) {3}};
            \draw[<->,>=stealth'] (1) -- (3);
            
            \node[node, below left= 6mm and 19mm of q3] (5) {5}
                child[bi_dir_e]{ node [node] (4) {4}}
                child[bi_dir_e]{ node [node] (6) {6}};
            \draw[<->,>=stealth'] (4) -- (6);
            
            \node[circle, draw=black, fit=(1) (2) (3), dashed, below right=6mm and 10mm of q3] (q2) {};
            \node[circle, draw=black, fit=(4) (5) (6), dashed, below left=6mm and 10mm of q3] (q1) {};
            \draw[->,>=stealth'] (q1) -- (q3);
            \draw[->,>=stealth'] (q2) -- (q3);
        \end{tikzpicture}
        \caption{The strongly connected components of the FBAS from Example~\ref{ex:FBAS_fig7}.}\label{fig:FBAS_fig7}  
    \end{figure}
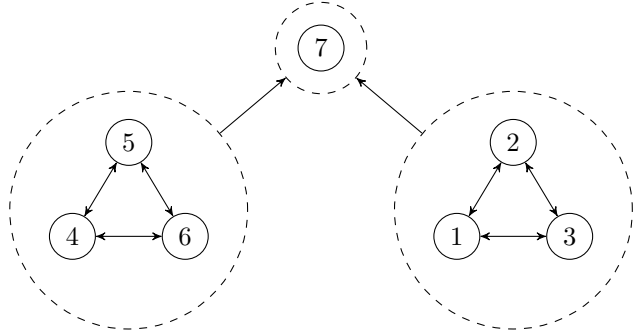
\end{ex}

The following result follows from the finiteness of the set $V$.

\begin{lem}\label{lm:maximal-SCC}
    There is at least one maximal SCC. If there is only one maximal SCC, then it is the greatest SCC. The greatest SCC is also maximal. 
\end{lem}

Let us now introduce a generalization of maximal SCCs.

\begin{definition}[trust cluster]\label{def:trust-cluster}
    A set $Z\subseteq V$ is called a \emph{trust cluster} of $(V,S)$ if it is closed under reachability in $G$, i.e., if for every $u\in Z$ and $v\in V$ reachable from $u$, also $v\in Z$.
\end{definition}

\begin{lem}\label{lem:trust-cluster}
    A set $Z\subseteq V$ is a trust cluster if and only if it is the union of strongly connected components that are closed under reachability. In particular, every maximal SCC is a trust cluster.
\end{lem}

If a node $v$ is in a trust cluster, then the cluster contains every node that $v$ trusts transitively. Moreover, we have the following basic properties.

\begin{lem}\label{lm:trust-cluster}
    Let $Z$ be a trust cluster of $(V,S)$.
    \begin{enumerate}[(1)]
        \item $(Z, S|_Z)$ is a well-defined FBAS, where $S|_Z$ is the restriction of $S$ to $Z$.
        \item Every quorum of $(Z, S|_Z)$ is also a quorum of $(V,S)$.
        \item $Z$ is a quorum of $(V,S)$.
    \end{enumerate}
\end{lem}
\begin{proof}
    (1)~Let $v\in Z$. Then every quorum slice of $v$ is contained in $Z$, i.e., $S(v)\subseteq \mathcal P(Z)$ or, equivalently, $S(v)\in \mathcal P(\mathcal P(Z))$.
    
    (2)~Follows from the equality $S|_Z(v)=S(v)$ for every $v\in Z$.
    
    (3)~By~(2) and the fact that $Z$ is a quorum of $(Z, S|_Z)$.
\end{proof}
Note that Lemma~\ref{lm:trust-cluster} also applies if $Z$ is a maximal or the greatest SCC because of Lemma~\ref{lem:trust-cluster}.

\subsection{The Stellar network}

\emph{Stellar Core} 
uses hierarchical quorum slices, but the construction is more sophisticated than in Example~\ref{ex:hierarchy}. In order to describe this construction we need the following definition, which is based on a \emph{rooted tree}.

\begin{definition}[quorum slice definition]
    A \emph{quorum slice definition} $d=(t,N,C)$ consists of a threshold $t\in\mathbb{N}$, a set of nodes $N$, and a set of children definitions $C$, which satisfy the following conditions:
    \begin{enumerate}
        \item Each child $c\in C$ is a quorum slice definition and the recursively induced graph is a rooted tree with root $d$ and a finite number of vertices.
        \item For two vertices $(t_1,N_1,C_1)\neq(t_2,N_2,C_2)$ of the induced tree, the nodes do not intersect, i.e., $N_1\cap N_2=\emptyset$.
        \item $0\leq t \leq |N| + |C|$.
    \end{enumerate}
\end{definition}

For describing how a quorum slice definition is translated to quorum slices we need the following definitions.

\begin{definition}[$k$-subsets]\label{def:k-subsets}
For a finite set $M$ and an integer $k\leq |M|$ we define the $k$-subsets of $M$ by $\mathcal{P}_k(M)=\{N\subseteq M\mid |N|=k\}\subseteq \mathcal P(M)$.
\end{definition}

\begin{definition}[product set union]
For sets of sets $M_1,\ldots,M_k$ we define their product set union by
    \begin{align*}
      \cup_\times(M_1,\ldots,M_k)
        &=\{m_1\cup\cdots\cup m_k\mid (m_1,\ldots,m_k)\in M_1\times\ldots\times M_k\}\\ 
        &\subseteq \mathcal P\left( \bigcup (M_1\cup \cdots \cup M_k)\right).
    \end{align*}
\end{definition}

\begin{definition}[quorum slice generation]
Let $d=(t,N,C)$ be a quorum slice definition with threshold $t$, nodes $N=\{n_1,\ldots,n_l\}$, and children definitions $C=\{c_1,\ldots,c_m\}$. The quorum slices $s(d)$ are generated recursively with the following set of $l+m$ quorum slices sets $$B_{N,C}=\bigl\{\{\!\{n_1\}\!\},\ldots,\{\!\{n_l\}\!\},s(c_1),\ldots,s(c_m)\bigr\}$$
by considering all $t$-subsets of $B_{N,C}$ and reducing them to a single set of quorum slices with the product set union:
$$s(d)=\bigcup \{\cup_\times (M_1,\ldots,M_t) \mid \{M_1,\ldots M_t\}\in \mathcal{P}_t(B_{N,C})\}.$$
\end{definition}

Note that the above recursion is well-defined and finite because the tree has by definition a finite number of vertices and thus leaf vertices, i.e., where the set of children definitions is empty.

In \emph{Stellar Core} a quorum slice definition is transformed with respect to the node it is running so that the node itself is always contained in the generated quorum slices. For a formal description we also need the following definition.

\begin{definition}[removal of a node]\label{def:node-removal}
Let $d=(t,N,C)$ be a quorum slice definition and let $v$ be a node. The quorum slice definition with $v$ removed is defined by
$$ R_v(d)=
    \begin{cases}
        (t-1,N\setminus \{v\},C),&\text{if $v\in N$},\\
        (t,N,\{R_v(c)\mid c\in C\}),&\text{otherwise}.
    \end{cases}
$$
\end{definition}

A \emph{Stellar Core} instance running on node $v$ that is given a quorum slice definition $d$ actually uses the quorum slice definition $(2,\{v\},\{R_v(d)\})$.

\begin{ex}[Stellar network]\label{ex:stellar-core}
    The current Stellar network 
    consists of 20 nodes in its greatest SCC (see Definition~\ref{def:SCC-props}), which are run by 6 organizations\footnote{\emph{Network transitive quorumset} taken from \url{https://stellarbeat.io} on Nov. 7, 2019.}. 
    The nodes are $V=A\cup B\cup C\cup D\cup E\cup F$ with
    \begin{align*}
        A&=\{a_1,a_2,a_3\},\quad
        B=\{b_1,b_2,b_3\},\quad
        C=\{c_1,c_2,c_3\},\\
        D&=\{d_1,d_2,d_3\},\quad
        E=\{e_1,e_2,e_3\},\quad
        F=\{f_1,f_2,f_3,f_4,f_5\}.
    \end{align*}
    The base quorum slice definition is
    $$d=(5,\emptyset,\{(2,A,\emptyset),\ldots,(2,E,\emptyset),(3,F,\emptyset)\}).$$
    Thus, all quorum slices need to cover 5 organizations, and within each organization two nodes must be present (except for $F$, which requires three nodes to be present). The quorum slices of $v\in V$ are given by 
    $$S(v)=s(d_v),$$ where $d_v=(2,\{v\},\{R_v(d)\})$ is the transformed quorum slice definition with $v$ moved to the top level. \\
    The nodes in $V\setminus F$ have \numprint{3132} quorum slices and the nodes in $F$ have \numprint{2673} quorum slices.
    A computation in \emph{Stellar Observatory} shows that the FBAS $(V,S)$ has \numprint{114688} quorums, and it has quorum intersection with the minimal size of an intersection given by 4 nodes. 
\end{ex}
\medskip

\section{Algorithmic treatment of quorums}\label{sec:alg-quorums}

In this section we will first show that the quorum intersection decision problem is NP-complete, and we will then discuss algorithms for enumerating the quorums of a given FBAS and for deciding whether an FBAS has quorum intersection.

\subsection{Checking quorum intersection is NP-complete}\label{sec:NPcomplete}

Lachowski~\cite{Lac19} showed that the problem of deciding whether a given FBAS $(V,S)$ has two disjoint quorums (Disjoint Quorums Problem) is NP-complete. His proof relies on a polynomial-time reduction from the NP-complete Set Splitting Problem~\cite{Garey79} to the Disjoint Quorums Problem. Given an instance of the Set Splitting Problem (a family $F$ of subsets of a given set $S$), his reduction builds an instance of the Disjoint Quorums Problem such that it is only possible to partition the set $S$ into subsets $S_1, S_2$ in a way that all elements of $F$ are split by this partition, if and only if there are disjoint quorums in the constructed FBAS.

Here we will show that the decision problem is NP-complete even if we restrict it to the class of \emph{simple} FBAS; see Definition~\ref{def:simple}. In order to fix the terminology, we say that an FBAS $(V,S)$ has \emph{quorum split} if there exist two disjoint quorums, and we let $\textsc{SimpleQuorumSplit}$ be the problem to decide whether a given \emph{simple} FBAS does have quorum split.   

\begin{thm}\label{thm:quorumsplit-np}
    \textsc{SimpleQuorumSplit} is NP-complete.
\end{thm}
\begin{proof}
	\textsc{SimpleQuorumSplit} is obviously in NP because it can be checked in polynomial time whether a set is a quorum (see Lemma~\ref{lm:isQuorum} below) and, therefore, whether two witness sets are indeed disjoint quorums.

	It remains to show that \textsc{SimpleQuorumSplit} is NP hard. We provide a polynomial-time reduction of the \textsc{3-SAT} \cite{Kar72} problem to \textsc{SimpleQuorumSplit}.
	Consider an instance $\varphi = (l_1^1\vee l_1^2\vee l_1^3)\wedge \cdots \wedge (l_m^1\vee l_m^2\vee l_m^3)$ of \textsc{3-SAT}, where the literals
	$l_i^j$ are either variables or negated variables taken from the set $x_1, \ldots, x_r$ of variables. We construct a simple FBAS $(V, q, n)$ that has quorum split if and only if $\varphi$ is satisfiable. Let
    \begin{align}
        V &=\{\star\}\cup\{a_k, x_k, \neg x_k\mid 1\leq k\leq r\}\cup\{b_i, c_i^1, c_i^2, c_i^3\mid 1\leq i\leq m\},\nonumber\\
        q(\star) &=\{\star\}\cup \{b_i\mid 1\leq i\leq m\}\;\mbox{and $n(\star)=m + 1$,}\label{proof:quorumsplit-np-def-star}\\
        q(a_k) &=\{a_k,x_k,\neg x_k\}\;\mbox{and $n(a_k)=2$ for $1\leq k\leq r$},\label{proof:quorumsplit-np-def-ak}\\
        q(l) &=\bigl\{l, a_{(k\bmod r)+1}\bigr\}\cup\bigl\{c_i^j\mid l_i^j = l, 1\leq i\leq m, 1\leq j\leq 3\bigr\}\; \mbox{and}\label{proof:quorumsplit-np-def-ql}\\
        n(l)& =|q(l)|\;\mbox{for $1\leq k\leq r$ and $l\in\{x_k, \neg x_k\}$},\\
        q(b_i) &=\{b_i, c_i^1, c_i^2, c_i^3\}\;\mbox{and $n(b_i)=2$ for 
        $1\leq i\leq m$,}\label{proof:quorumsplit-np-def-qb}\\
        q(c_i^j) &=\{c_i^j, \star, a_1\}\; \mbox{and $n(c_i^j)=2$ for $1\leq i\leq m$ 
        and $1\leq j\leq 3$.}\label{proof:quorumsplit-np-def-qc}
    \end{align}

	Obviously, $(V,q,n)$ is a valid simple FBAS and can be constructed in time polynomial in the size of $\varphi$. It remains to show that $(V, q, n)$ has quorum split if and only if $\varphi$ is satisfiable.

	``$\Leftarrow$'': Suppose that $\varphi$ is satisfiable and let $\alpha:\{x_1, \ldots, x_r\}\rightarrow \{\mathit{true}, \mathit{false}\}$ be a variable assignment that satisfies $\varphi$.
	Let
	\begin{align*}
		L=\{x_k\mid \alpha(x_k)=\mathit{false},\; 1\leq k\leq r\}\cup\{\neg x_k\mid \alpha(x_k)=\mathit{true},\; 1\leq k\leq r\}
	\end{align*}
	be the literals that are false under $\alpha$. Consider the two sets
	\begin{align*}
		Q_1 &=\{a_k\mid 1\leq k\leq r\}\cup L\cup\{c_i^j\mid l_i^j\in L, 1\leq i\leq m, 1\leq j\leq 3\},\\
		Q_2 &=\{\star\}\cup \{b_i\mid 1\leq i\leq m\}\cup\{c_i^j\mid l_i^j\not\in L, 1\leq i\leq m, 1\leq j\leq 3\}.
    \end{align*}
	Clearly, $Q_1$ and $Q_2$ are disjoint. Moreover, $Q_1$ and $Q_2$ are quorums of $(V,q,n)$ because for every $1\leq i\leq m, 1\leq j\leq 3$ and $1\leq k\leq r$ we have:
	\begin{itemize}
		\item $x_k\in L$ or $\neg x_k\in L$, and hence by~\eqref{proof:quorumsplit-np-def-ak} either $\{a_k, x_k\}$ or $\{a_k, \neg x_k\}$ is a quorum slice of $a_k$ contained in $Q_1$,
		\item $q(l)\subseteq Q_1$ for every $l\in L$ by~\eqref{proof:quorumsplit-np-def-ql}, and therefore every quorum slice of $l$ is contained in $Q_1$,
		\item if $l_i^j\in L$, then $\{c_i^j, a_1\}$ is a quorum slice of $c_i^j$ contained in $Q_1$ by~\eqref{proof:quorumsplit-np-def-qc},
		\item $\{\star, b_1, \ldots, b_m\}$ is a quorum slice of $\star$ contained in $Q_2$ by~\eqref{proof:quorumsplit-np-def-star},
		\item since $\alpha$ is a satisfying variable assignment for $\varphi$, there is a $1\leq j\leq 3$ such that $l_i^j\not\in L$; then $\{b_i, c_i^j\}$ is a quorum slice of $b_i$ contained in $Q_2$ by~\eqref{proof:quorumsplit-np-def-qb},
		\item if $l_i^j\not\in L$, then $\{c_i^j, \star\}$ is a quorum slice of $c_i^j$ contained in $Q_2$ by~\eqref{proof:quorumsplit-np-def-qc}.
	\end{itemize}

	``$\Rightarrow$'': Let $Q_1$ and $Q_2$ be disjoint quorums in $(V, q, n)$. We construct a satisfying variable assignment $\alpha$ for $\varphi$.
	First note that for every quorum $Q$ we have $\star\in Q$ or $a_1\in Q$, which follows from this list of observations:
	\begin{itemize}
		\item if $a_k\in Q$, then $x_k\in Q$ or $\neg x_k\in Q$ by~\eqref{proof:quorumsplit-np-def-ak},
		\item if $x_k\in Q$ or $\neg x_k \in Q$, then $a_{(k\bmod r)+1}\in Q$ by~\eqref{proof:quorumsplit-np-def-ql},
		\item if $b_i\in Q$, then $c_i^1\in Q$, $c_i^2\in Q$ or $c_i^3\in Q$ by~\eqref{proof:quorumsplit-np-def-qb},
		\item if $c_i^j\in Q$, then $\star \in Q$ or $a_1\in Q$ by~\eqref{proof:quorumsplit-np-def-qc}.
	\end{itemize}
	Since $Q_1$ and $Q_2$ are disjoint, we assume without loss of generality that $a_1\in Q_1$ and $\star\in Q_2$. We immediately have $\{b_1, \ldots, b_m\}\subseteq Q_2$ due to~\eqref{proof:quorumsplit-np-def-star}, and from the list of
	observations above we derive that $\{a_1, \ldots, a_r\}\subseteq Q_1$ and that, for every $1\leq k\leq r$, there is an $l_k \in \{x_k, \neg x_k\} \cap Q_1$. 
	
	Now set $\alpha(x_k)=\mathit{true}$ if $l_k=\neg x_k$, and set $\alpha(x_k)=\mathit{false}$ if $l_k=x_k$. In order to show that $\alpha$ satisfies $\varphi$, it
	suffices to show that for every $1\leq i\leq m$, there is a $1\leq j\leq 3$ such that $l_i^j$ is $\mathit{true}$ under $\alpha$. Since $Q_2$ is a quorum and $b_i\in Q_2$,
	there is a $1\leq j\leq 3$ such that $c_i^j\in Q_2$ by~\eqref{proof:quorumsplit-np-def-qb}. We show that $l_i^j$ is $\mathit{true}$ under~$\alpha$. Observe that $l_i^j\not\in Q_1$,
	otherwise we would have $c_i^j\in Q_1$ as $Q_1$ is a quorum and due to~\eqref{proof:quorumsplit-np-def-ql}, contradicting $c_i^j\in Q_2$. There is a $1\leq k\leq r$ with $l_i^j=x_k$ or $l_i^j=\neg x_k$.
	Then, $l_i^j\not=l_k$ as $l_k\in Q_1$, i.e., $l_k=\neg x_k$ if $l_i^j=x_k$, and $l_k=x_k$ if $l_i^j=\neg x_k$. By the definition of $\alpha$, $l_i^j$ is $\mathit{true}$ under $\alpha$.
\end{proof}

In the following two subsections we will discuss algorithms for
enumerating the quorums of a given FBAS and for deciding whether
an FBAS has quorum intersection. These algorithms are based on the work of Lachowski~\cite{Lac19}. Lachowski describes them only on a conceptual level which gives room for interpretation when implementing both algorithms. For this purpose we base our discussion additionally on the current implementation in \emph{Stellar Core}\/\footnote{\url{https://github.com/stellar/stellar-core/blob/b30a7023e6b51843ac2d3763d04ef7f0ad56ad73/src/herder/QuorumIntersectionCheckerImpl.h}}. We remark that the original implementation of the quorum intersection algorithm in \emph{Stellar Core} was erroneous\/\footnote{\url{https://github.com/stellar/stellar-core/issues/2267}}. The misinterpretation of Lachowski's account might be partially due to the fact that he merely gave a sketch of these algorithms. For these reasons our main purpose here is to give an in-depth and precise treatment of both algorithms.

\subsection{Quorum enumeration}\label{sec:quorum-enumeration}
    In order to describe the computational cost of the algorithms, we need to define a measure for the size of the input. For the purposes of this definition we distinguish between simple and general FBAS.
    
    \begin{definition}\label{def:fbas-size}
        The \emph{slice size} $\Vert v\Vert$ of a node $v$ is defined as
        \begin{enumerate}[(1)]
            \item $\Vert v\Vert = \sum\nolimits_{s\in S(v)}|s|$,
            for a general FBAS $(V,S)$, and
            \item $\Vert v\Vert = |q(v)|$, for a simple FBAS $(V, q, n)$.
        \end{enumerate}
        The \emph{size} of an FBAS $F=(V,S)$ or $F=(V,q,n)$, respectively, is defined as $\Vert F\Vert =|V| + \sum_{v\in V}\Vert v\Vert$.
    \end{definition}
    
    The most fundamental operation for all algorithms discussed in this section is to decide whether a given node has a quorum slice contained in a given subset of $V$, see function {\sc containsSlice} in Figure~\ref{fig:alg-contains-slice}.
    
    \begin{lem}\label{lm:alg-contains-slice}
        Given $U\subseteq V$ and $v\in U$ the function {\sc containsSlice} decides in time $O(\Vert v\Vert)$ whether $v$ has a quorum slice contained in $U$.
    \end{lem}
    \begin{proof}
       Figure~\ref{fig:alg-contains-slice} defines the function {\sc containsSlice} for general as well as for simple FBAS. The correctness of the former version is obvious. For the correctness of the latter version note that $U$ contains a slice of $v$ if and only if there are $n(v)$ nodes in $q(v)$ that are also in $U$ if and only if $q(v) \cap U$ has at least $n(v)$ elements.
       
       The complexity follows from the fact that $s\subseteq U$ can be decided in time $O(|s|)$ and that $U\cap q(v)$ can be computed in time $O(|q(v)|)$.
    \end{proof}
    
    \begin{figure}[tbh]
        \begin{center}
            \begin{tabular}{rl}
                \hline \multicolumn{2}{l}{{\sc containsSlice} (for general FBAS)}\\\hline
                \multicolumn{2}{l}{\emph{Input:} FBAS $(V, S)$; $U\subseteq V$; $v\in U$}\\
                \multicolumn{2}{l}{\emph{Output:} whether $v$ has a quorum slice contained in $U$}\\
                1&{\bf for} $s$ {\bf in} $S(v)$\\
                2&\qquad{\bf if} $s\subseteq U$ {\bf return} $\mathit{true}$\\
                3&{\bf return} $\mathit{false}$\\\hline\\
                \hline \multicolumn{2}{l}{{\sc containsSlice} (for simple FBAS)}\\\hline
                \multicolumn{2}{l}{\emph{Input:} simple FBAS $(V, q, n)$; $U\subseteq V$; $v\in U$}\\
                \multicolumn{2}{l}{\emph{Output:} whether $v$ has a quorum slice contained in $U$}\\
                1&{\bf return} $|U\cap q(v)|\geq n(v)$\\\hline
            \end{tabular}
        \end{center}
        \caption{Algorithm for the function {\sc containsSlice} for both general and simple FBAS.}\label{fig:alg-contains-slice}
    \end{figure}
    
    In the sequel we will not distinguish between general and simple FBAS anymore. To this end we will refer to any FBAS simply using the variable~$F$. We can now use the function {\sc containsSlice} to build more interesting algorithms, see Figure~\ref{fig:alg-is-quorum}.
    \begin{lem}\label{lm:isQuorum}
        Given a nonempty set $U\subseteq V$ the function {\sc isQuorum} decides in time $O(\Vert F\Vert)$ whether $U$ is a quorum of the FBAS $F$.
    \end{lem}
    \begin{proof}
        The correctness of {\sc isQuorum} follows immediately from the definition of quorums. It requires time $O(\sum_{v\in U}\Vert v\Vert)\leq O(\Vert F\Vert)$ due to Lemma~\ref{lm:alg-contains-slice}.
    \end{proof}
    
    \begin{figure}[tbh]
        \begin{center}
            \begin{tabular}{rl}
                \hline \multicolumn{2}{l}{{\sc isQuorum}}\\\hline
                \multicolumn{2}{l}{\emph{Input:} FBAS $F$ with set of nodes $V$; $U\subseteq V$ nonempty}\\
                \multicolumn{2}{l}{\emph{Output:} whether $U$ is a quorum in $F$}\\
                1&{\bf for} $v$ {\bf in} $U$\\
                2&\qquad{\bf if not} $\textsc{containsSlice}(F, U, v)$ {\bf return} $\mathit{false}$\\
                3&{\bf return} $\mathit{true}$\\\hline\\
                \hline \multicolumn{2}{l}{{\sc greatestQuorum}}\\\hline
                \multicolumn{2}{l}{\emph{Input:} FBAS $F$ with set of nodes $V$; $U\subseteq V$, $W\subseteq U$ ($\emptyset$ by default)}\\
                \multicolumn{2}{l}{\emph{Output:} greatest quorum $Q$ with $W\subseteq Q\subseteq U$ if it exists, $\emptyset$ otherwise}\\
                1&$U_1\leftarrow U$, $i\leftarrow 1$\\
                2&{\bf repeat}\\
                3&\qquad$U_{i+1}\leftarrow\emptyset$\\
                4&\qquad{\bf for} $v$ {\bf in} $U_i$\\
                5&\qquad\qquad{\bf if} $\textsc{containsSlice}(F, U_i, v)$\\
                6&\qquad\qquad\qquad $U_{i+1}\leftarrow U_{i+1}\cup\{v\}$\\
                7&\qquad\qquad{\bf else}\\
                8&\qquad\qquad\qquad{\bf if} $v \in W$ {\bf return} $\emptyset$\\
                9&\qquad{\bf if} $U_{i+1}=U_{i}$ {\bf or} $U_{i+1}=\emptyset$ {\bf return} $U_{i+1}$\\
                10&\qquad$i\leftarrow i+1$\\\hline
            \end{tabular}
        \end{center}
        \caption{Algorithm for {\sc isQuorum} and {\sc greatestQuorum}.}
        \label{fig:alg-is-quorum}
    \end{figure}
    
    Clearly the union of two quorums is again a quorum. Thus, if a set $U\subseteq V$ contains at least one quorum, then the union of all quorums contained in $U$ is again a quorum. This quorum is the greatest with respect to the subset relation among all quorums contained in $U$. The following lemma states that the greatest quorum contained in a set of nodes can be computed efficiently via the function {\sc greatestQuorum}, see Figure~\ref{fig:alg-is-quorum}. This function is slightly more general: given two sets $U\subseteq V$ and $W\subseteq U$ it determines the greatest quorum $Q$ with $W\subseteq Q\subseteq U$. Note that we specify this last argument $W$ to be the empty set by default whenever we do not provide it explicitly when calling {\sc greatestQuorum} in other algorithms defined in the remainder of this paper.
    
    \begin{lem}\label{lm:alg-greatest-quorum}
        Given sets $U\subseteq V$ and $W\subseteq U$ the function {\sc greatestQuorum} decides whether there is a quorum $Q$ with $W\subseteq Q\subseteq U$ and, if so, computes the greatest such quorum in time $O(|V|\cdot \Vert F\Vert)$.
    \end{lem}
    \begin{proof}
        Observe that the $i$-th iteration of the loop in lines~2--10 constructs the set $U_{i+1}$. The construction is performed in lines~3--8. Whenever we refer to one of the sets $U_{i+1}$ in this proof, we refer to its final value that it attains after its construction.
    
        First observe that $U_{i+1}\subseteq U_i$ for every $i$ by construction of $U_{i+1}$. Hence the loop in lines~2--10 will terminate in line~9 after at most $|U|\leq |V|$ iterations. By Lemma~\ref{lm:alg-contains-slice} each iteration takes time $O(\sum_{v\in U_i}\Vert v\Vert + |U_i|)\leq O(\Vert F\Vert)$. This proves the time complexity of {\sc greatestQuorum}.
        
        The set returned by {\sc greatestQuorum} is either empty or a quorum contained in $U$ because the condition $U_{i+1}=U_i$ on line~9 is equivalent to the condition that $\textsc{containsSlice}(F, U_i, v)$ returns $\mathit{true}$ for every $v\in U_i$ -- compare with the function {\sc isQuorum}. Furthermore, if the algorithm does not return the empty set in line~8, then it is easy to show by induction that $W$ is contained in every constructed $U_i$: (i)~it is obvious for $U_1$ and (ii)~for every $U_{i+1}$ and $v\in W$ the element $v$ is added to $U_{i+1}$ in line~6. We conclude that if the algorithm does not return the empty set, then it returns a quorum contained in $U$ that contains $W$.
        Conversely, whenever $U$ does not contain a quorum that contains $W$, then {\sc greatestQuorum} returns the empty set as expected.
        
        Finally, let us consider the case that $U$ contains a quorum that contains $W$. Given an arbitrary quorum $Q$ with $W\subseteq Q\subseteq U$ we will show by induction that $Q\subseteq U_i$ for every constructed $U_i$ and that the algorithm does not return in line~8. This implies that {\sc greatestQuorum} returns a quorum that is a superset of $Q$; since $Q$ is arbitrary, {\sc greatestQuorum} returns even the greatest quorum in $U$ that contains $W$. Now let us carry out the inductive proof. Trivially, $Q\subseteq U_1=U$. Let $i$ with $Q\subseteq U_i$. For every $v\in Q$ there is a quorum slice of $v$ contained in $Q$ and therefore also in $U_i$ -- hence, $v\in U_{i+1}$. We conclude $Q\subseteq U_{i+1}$.
    \end{proof}

    Let us now consider an algorithm for enumerating all quorums of an FBAS. Note that every such algorithm has worst-case time complexity of at least $O(2^{|V|})$ because there can be exponentially many quorums. In fact, there are FBAS such that every nonempty subset of $V$ is a quorum; see Example~\ref{ex:trivial}.
    
    We will show that quorums can be enumerated with polynomial-time delay using the function {\sc enumerateQuorums} in Figure~\ref{fig:alg-quorum-enumeration}. This result is originally due to Lachowski~\cite{Lac19}, but the algorithm we present here uses a more efficient approach than the one described in his proof  of~\cite[Proposition~2]{Lac19}.
    \begin{lem}
        The function {\sc enumerateQuorums} enumerates all quorums of an FBAS with time delay $O(|V|^2\cdot\Vert F\Vert)$.
    \end{lem}
    \begin{proof}
        It suffices to show that for every $U\subseteq V$ and $R\subseteq V\setminus U$, the function call $\textsc{traverseQuorums}(F, U, R)$ enumerates all quorums $Q$ of $F$ such that $U\subseteq Q\subseteq U\cup R$ with time delay $O(|V|^2\cdot\Vert F\Vert)$.

        First let us prove the correctness of the algorithm by strong induction on $|R|$. Let $U\subseteq V$ and $R\subseteq V\setminus U$ and suppose that the correctness of {\sc traverseQuorums} holds for every $U'$ and $R'$ disjoint from $U'$ with $|R'|< |R|$.

        Consider the case that the function call returns early in line~2 and therefore does not enumerate any quorums. Then $Q=\emptyset$; by Lemma~\ref{lm:alg-greatest-quorum} this implies that $U\cup R$ does not contain any quorums that are supersets of $U$ and therefore there are no quorums to enumerate -- the function behaves correctly.

        Now consider the case that the algorithm does not return early in line~2. Then $Q$ is a quorum with $U\subseteq Q\subseteq U\cup R$ and line~3 will correctly output this quorum. All remaining quorums $Q'$ with $U\subseteq Q'\subseteq U\cup R$ will be proper subsets of $Q$ as $Q$ is the greatest quorum contained in $U\cup R$. Therefore, in order to enumerate all remaining quorums it is sufficient to consider candidates in the set $\{Q'\mid U\subseteq Q'\subset Q\}$. Suppose that the set $Q\setminus U$ has $n$ elements and that they are enumerated in the order $v_1, \ldots, v_n$ in the loop in lines~5 to~7. Observe that we can decompose the set $\{Q'\mid U\subseteq Q'\subset Q\}$ disjointly as follows:
        \begin{align}
            \{Q'\mid U\subseteq Q'\subset Q\} &= \big\{Q'\mid U\subseteq Q'\subseteq Q\setminus\{v_1\}\big\}\nonumber\\
            &\phantom{=\ }\mbox{}\cup \big\{Q'\mid U\cup\{v_1\}\subseteq Q'\subseteq Q\setminus\{v_2\}\big\}\nonumber\\
            &\phantom{=\ }\mbox{}\cup \cdots\cup\big\{Q'\mid U\cup\{v_1, \ldots, v_{n-1}\}\subseteq Q'\subseteq Q\setminus\{v_n\}\big\}\nonumber\\
            &=\bigcup\nolimits_{i=1}^n\big\{Q'\mid U\cup\{v_1, \ldots, v_{i-1}\}\subseteq Q'\subseteq Q\setminus\{v_i\}\big\}.\label{eq:proof-enum-quorums}
        \end{align}
        The $i$-th component of this union stands for all sets $Q'$ that contain the nodes $v_1, \ldots, v_{i-1}$, but not the node $v_i$. Thus,~\eqref{eq:proof-enum-quorums} is a disjoint union. Let $W_i$ be the value of $W$ at the beginning of the $i$-th iteration of the loop in lines~5 to~7. Clearly, $W_i=\{v_{i}, \ldots, v_n\}$. Then the above decomposition can be written as
        \begin{align}
            \{Q'\mid U\subseteq Q'\subset Q\} &= \bigcup\nolimits_{i=1}^n\big\{Q'\mid Q\setminus W_i\subseteq Q'\subseteq Q\setminus\{v_i\}\big\}.\label{eq:proof-enum-quorums2}
        \end{align}
        Together with the induction hypothesis this proves the correctness of the algorithm because the $i$-th call to {\sc traverseQuorums} in line~6 will enumerate the quorums in the $i$-the component of the right-hand side of~\eqref{eq:proof-enum-quorums2}.

        For the proof of the time complexity note that every call to the function {\sc traverseQuorums} will either return early or output its first quorum after time $O(|V|\cdot \Vert F\Vert)$ by Lemma~\ref{lm:alg-greatest-quorum}. If it outputs a quorum, then it will make at most $|V|$ subsequent recursive calls to {\sc traverseQuorums}. Therefore the call to {\sc traverseQuorums} will either quit without outputting any other quorums, or it will output the next quorum after at most $O(|V|^2\cdot \Vert F\Vert)$ steps.
    \end{proof}

    \begin{figure}[tbh]
        \begin{center}
            \begin{tabular}{rl}
                \hline \multicolumn{2}{l}{{\sc enumerateQuorums}}\\\hline
                \multicolumn{2}{l}{\emph{Input:} FBAS $F$ with set of nodes $V$}\\
                \multicolumn{2}{l}{\emph{Output:} enumerate all quorums of $F$}\\
                1&$\textsc{traverseQuorums}(F, \emptyset, V)$\\\hline\\
                \hline \multicolumn{2}{l}{{\sc traverseQuorums}}\\\hline
                \multicolumn{2}{l}{\emph{Input:} FBAS $F$ with set of nodes $V$; $U\subseteq V$; $R\subseteq V\setminus U$}\\
                \multicolumn{2}{l}{\emph{Output:} enumerate all quorums $Q$ of $F$ with $U\subseteq Q\subseteq U\cup R$}\\
                1&$Q\leftarrow\textsc{greatestQuorum}(F, U\cup R, U)$\\
                2&{\bf if} $Q=\emptyset$ {\bf return}\\
                3&{\bf output} $Q$\\
                4&$W\leftarrow Q\setminus U$\\
                5&{\bf for} $v$ {\bf in} $Q\setminus U$\\
                6&\qquad$\textsc{traverseQuorums}(F, Q\setminus W, W\setminus\{v\})$\\
                7&\qquad$W\leftarrow W\setminus\{v\}$\\\hline
            \end{tabular}
        \end{center}
        \caption{Algorithm for {\sc enumerateQuorums}.}
        \label{fig:alg-quorum-enumeration}
    \end{figure}
    
    \begin{ex}
        \label{ex:quorum-enumeration-computation}
        In this experiment we study the computation time of the {\sc enumerateQuorums} algorithm (see Figure~\ref{fig:alg-quorum-enumeration}) implemented in \emph{Stellar Observatory} in comparison with the number of quorums of a sequence of FBAS that are constructed similarly to the current maximal SCC of the Stellar network; see Example~\ref{ex:stellar-core}. We consider $n$ organizations $A_1,\ldots,A_n$ with $3$ nodes each, and define an FBAS $(V,S)$ 
        by $V=A_1\cup\cdots\cup A_n$, so that $|V|=3n$, and quorum slices $$S(v)=s(d_v),$$ where $d_v=(2,\{v\},\{R_v(d)\})$ and 
        $d=(t,\emptyset,\{(2,A_1,\emptyset),\ldots,(2,A_n,\emptyset)\})$ for a $t\in\N$.
        Essentially, $t$ organizations are required on the root level and then $2$ nodes are required within each organization. Analogous to the Stellar network (cf.~Example~\ref{ex:stellar-core}), we transform the quorum slice definition for each node $v\in V$ to $d_v$ by moving $v$ to the root level and removing it from its original position with $R_v$; see Definition~\ref{def:node-removal}.
        
        We consider the root level thresholds $t=n-1$ and $t=\lfloor\frac{2}{3}n\rfloor+1$. In both cases the number of quorums as well as the computation time of the {\sc enumerateQuorums} algorithm grow exponentially with $n$; see the dashed line~\eqref{fig:enumerate-quorums-1:quorums} 
        and the solid line~\eqref{fig:enumerate-quorums-1:time} in Figure~\ref{fig:quorum-enumeration-computation}, respectively. The code for reproducing the results is available on GitHub\footnote{\url{https://github.com/andrenarchy/stellar-experiments}}.
    \end{ex}

    \begin{figure}[htbp]
        \centering
        \setlength{\figurewidth}{0.75\textwidth}
        \setlength{\figureheight}{0.65\textwidth}
        \begin{subfigure}[b]{\textwidth}
            \centering

\begin{tikzpicture}

\begin{axis}[
width=\figurewidth,
height=\figureheight,
log basis y={10},
tick align=outside,
tick pos=left,
x grid style={white!69.01960784313725!black},
xlabel={Number of organizations \(\displaystyle n\)},
xmin=2.65, xmax=10.35,
xtick style={color=black},
y grid style={white!69.01960784313725!black},
ylabel={Computation time (s)},
ymin=0.001, ymax=10000,
ymode=log,
ytick style={color=black}
]
\addplot [thick, plot_color_0, mark=*, mark size=3, mark options={solid}]
table {%
3 0.00344014167785645
4 0.0254566669464111
5 0.173870086669922
6 1.17016506195068
7 6.65367412567139
8 40.9504787921906
9 273.812614440918
10 1218.99195337296
};
\label{fig:enumerate-quorums-1:time}
\end{axis}

\begin{axis}[
width=\figurewidth,
height=\figureheight,
axis y line=right,
log basis y={10},
tick align=outside,
x grid style={white!69.01960784313725!black},
xmin=2.65, xmax=10.35,
xtick pos=left,
xtick style={color=black},
y grid style={white!69.01960784313725!black},
ylabel={Number of quorums},
ymin=100, ymax=100000000,
ymode=log,
ytick pos=right,
ytick style={color=black}
]
\addplot [thick, plot_color_1, dotted, mark=*, mark size=3, mark options={solid}]
table {%
3 256
4 1280
5 6144
6 28672
7 131072
8 589824
9 2621440
10 11534336
};
\label{fig:enumerate-quorums-1:quorums}
\end{axis}

\end{tikzpicture}
            \caption{Root threshold $n-1$}
        \end{subfigure}
        \par\bigskip
        \begin{subfigure}[b]{\textwidth}
            \centering
\begin{tikzpicture}

\begin{axis}[
width=\figurewidth,
height=\figureheight,
log basis y={10},
tick align=outside,
tick pos=left,
x grid style={white!69.01960784313725!black},
xlabel={Number of organizations \(\displaystyle n\)},
xmin=2.65, xmax=10.35,
xtick style={color=black},
y grid style={white!69.01960784313725!black},
ylabel={Computation time (s)},
ymin=0.001, ymax=100000,
ymode=log,
ytick style={color=black}
]
\addplot [thick, plot_color_0, mark=*, mark size=3, mark options={solid}]
table {%
3 0.00114274024963379
4 0.0265226364135742
5 0.177913665771484
6 1.15317392349243
7 21.6138091087341
8 142.821984052658
9 880.701012849808
10 14708.7879264355
};
\end{axis}

\begin{axis}[
width=\figurewidth,
height=\figureheight,
axis y line=right,
log basis y={10},
tick align=outside,
x grid style={white!69.01960784313725!black},
xmin=2.65, xmax=10.35,
xtick pos=left,
xtick style={color=black},
y grid style={white!69.01960784313725!black},
ylabel={Number of quorums},
ymin=10, ymax=1000000000,
ymode=log,
ytick pos=right,
ytick style={color=black}
]
\addplot [thick, plot_color_1, dotted, mark=*, mark size=3, mark options={solid}]
table {%
3 64
4 1280
5 6144
6 28672
7 475136
8 2424832
9 12058624
10 184549376
};
\end{axis}

\end{tikzpicture}
            \caption{Root threshold $\lfloor\frac{2}{3}n\rfloor+1$}
        \end{subfigure}
        
        \caption{Computation time of {\sc enumerateQuorums} in seconds~\eqref{fig:enumerate-quorums-1:time} and number of quorums~\eqref{fig:enumerate-quorums-1:quorums} 
        for the FBAS described in Example~\ref{ex:quorum-enumeration-computation} consisting of $n$ organizations with $3$ nodes each. Note that the y-axes are scaled logarithmically.}
        \label{fig:quorum-enumeration-computation}
    \end{figure}
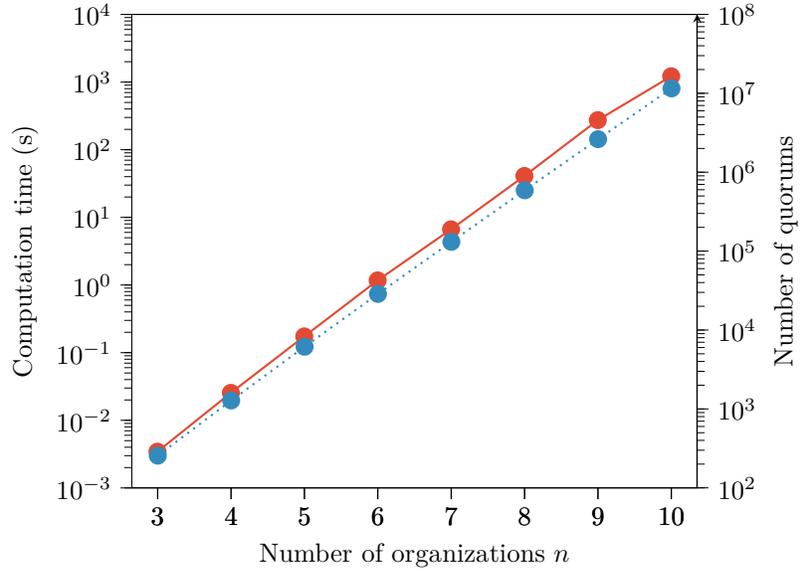
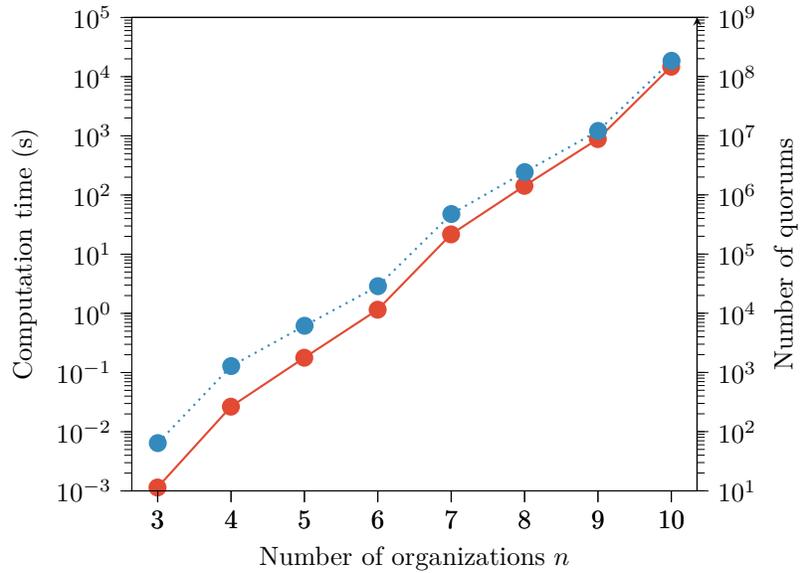

\subsection{Quorum intersection}\label{sec:quorum-intersection}

    We have shown in Section~\ref{sec:NPcomplete} that the problem to decide whether a given FBAS has quorum intersection is NP-complete. The brute-force approach to solve this problem is to first enumerate all quorums (see the previous section), and then to compare all pairs of quorums. This method has complexity $2^{O(|V|)}$ as the number of quorums can be exponential in the number of nodes.
    
    In this section we will discuss an algorithm to decide quorum intersection that still has exponential time complexity but that is drastically more efficient than the brute-force approach. This algorithm is due to Lachowski~\cite{Lac19} and is based on the following key ideas:
    \begin{enumerate}[(1)]
        \item Preprocessing the given FBAS can considerably reduce its size before enumerating the quorums.
        \item Instead of enumerating all quorums it is sufficient to enumerate only minimal quorums up to a certain size.
    \end{enumerate}

\subsubsection{Strongly connected components}\label{sec:quorum-intersection-scc}
    In this section we will make use of the trust graph as defined in Definition~\ref{def:trust-graph}. In the following lemma we use the notation $S|_C$, which is the restriction of $S$ to $C$; see Lemma~\ref{lm:trust-cluster}.
    \begin{lem}\label{lm:SCCs}
        \begin{enumerate}[(1)]
            \item If $(V, S)$ has quorum intersection, then the greatest SCC of its trust graph exists.
            \item Let $Q\subseteq V$ be a quorum of $(V,S)$. Then there is an SCC $D$ such that $D\cap Q$ is a quorum of $(V,S)$.
            \item If $C$ is the greatest SCC and no other SCC contains a quorum of $(V,S)$, then $(C, S|_C)$ has quorum intersection if and only if $(V,S)$ has quorum intersection.
        \end{enumerate}
    \end{lem}
    \begin{proof}
        (1)~If there are distinct maximal SCC, then there are disjoint quorums due to Lemma~\ref{lm:trust-cluster}(3). Hence, there is only one maximal SCC and it is even the greatest SCC by Lemma~\ref{lm:maximal-SCC}.
        
        (2)~Let $\mathcal D$ be the set of all SCCs that intersect $Q$. Then $\mathcal D$ is nonempty as $Q$ is nonempty. There is at least one SCC $D\in\mathcal D$ such that no other SCC in $\mathcal D$ is reachable from $D$. We show that $D\cap Q$ is a quorum of $(V,S)$. Let $v\in D\cap Q$. Then there is a slice $s\in S(v)$ with $s\subseteq Q$. Suppose that $s\not\subseteq D$ and let $u\in s\setminus D$. Thus, $u$ is in an SCC $D'$ different from $D$. Clearly, $D'\in \mathcal D$ because $u\in s\subseteq Q$. It follows that $D'$ is reachable from $D$, contrary to the construction of $D$. Thus, our assumption $s\not\subseteq D$ was wrong. Since $s\subseteq D\cap Q$, the set $D\cap Q$ is a quorum of $(V,S)$.
        
        (3) ``$\Rightarrow$'': Suppose that $(V,S)$ has disjoint quorums $Q_1$ and $Q_2$. Because of (2), there are SCCs $D_1$ and $D_2$ such that $D_1\cap Q_1$ and $D_2\cap Q_2$ are quorums of $(V,S)$. Then $D_1=D_2=C$ by assumption. The definition of $(C, S|_C)$ immediately implies that $C\cap Q_1$ and $C\cap Q_2$ are not just quorums of $(V,S)$ but also of $(C, S|_C)$. Therefore $(C, S|_C)$ does not have quorum intersection.
        
        ``$\Leftarrow$'': If $Q_1$ and $Q_2$ are disjoint quorums of $(C, S|_C)$, then they are also disjoint quorums of $(V,S)$ by Lemma~\ref{lm:trust-cluster}(2).
    \end{proof}
    
    Using the results of Lemma~\ref{lm:SCCs} we can preprocess the FBAS and reduce its size before running the actual algorithm to decide quorum intersection. To this end we determine the strongly connected components of the trust graph using, e.g., Tarjan's algorithm~\cite{Tar72}. If there is no greatest SCC, then Lemma~\ref{lm:SCCs}(1) implies that there are disjoint quorums. Next, we can check whether any of the SCCs different from the greatest SCC $C$ contain a quorum by using the function {\sc greatestQuorum}; see Lemma~\ref{lm:alg-greatest-quorum}. If this is the case, then there are disjoint quorums because the greatest SCC is itself a quorum by Lemma~\ref{lm:trust-cluster}(3). Otherwise, Lemma~\ref{lm:SCCs}(3) allows to reduce the quorum intersection decision to the FBAS $(C, S|_C)$.

\subsubsection{Minimal quorums}
    Instead of enumerating all quorums of $(V,S)$ and comparing them pairwise, we can instead enumerate all minimal quorums $Q$ up to size $|V|/2$ and check whether $V\setminus Q$ contains a quorum; see the function {\sc quorumIntersection} in Figure~\ref{fig:alg-quorum-intersection}. For now assume that $\textsc{traverseMinQuorums}(F, \emptyset, V)$ enumerates all minimal quorums up to size $|V|/2$; we will prove its correctness in Lemma~\ref{lm:traverseMinQuorums}.
    
    \begin{lem}
        The function {\sc quorumIntersection} decides whether $F$ has quorum intersection and returns a pair of disjoint quorums if it does not.
    \end{lem}
    \begin{proof}
         Suppose that {\sc quorumIntersection} returns a pair $(Q,Q')$. Thus $Q$ is a minimal quorum and $Q'=\textsc{greatestQuorum}(F, V\setminus Q)\not=\emptyset$. Then by Lemma~\ref{lm:alg-greatest-quorum}, $Q'$ is a quorum which is obviously disjoint from $Q$. Hence, $(V,S)$ does not have quorum intersection.
        
        Now suppose that $(V,S)$ does not have quorum intersection. We will prove that {\sc quorumIntersection} does not return $\mathit{true}$ but two disjoint quorums instead. There are two disjoint quorums $Q_1$ and $Q_2$. Without loss of generality, $|Q_1|\leq |Q_2|$, which implies $|Q_1|\leq |V|/2$. There is a minimal quorum $Q$ in $Q_1$, possibly $Q_1$ itself. Clearly, $Q$ will be one of the quorums enumerated by $\textsc{traverseMinQuorums}(F, \emptyset, V)$ in line~1. Since $Q_2$ is a quorum contained in $V\setminus Q$, the set $Q'$ computed in line~2 will not be empty and therefore {\sc quorumIntersection} will not return $\mathit{true}$.
    \end{proof}
    
    \begin{figure}[tb]
        \begin{center}
            \begin{tabular}{rl}
                \hline \multicolumn{2}{l}{{\sc quorumIntersection}}\\\hline
                \multicolumn{2}{l}{\emph{Input:} FBAS $F$ with set of nodes $V$}\\
                \multicolumn{2}{l}{\emph{Output:} $\mathit{true}$ if $F$ has quorum intersection; otherwise two disjoint quorums}\\
                1&{\bf for} $Q$ {\bf in} $\textsc{traverseMinQuorums}(F, \emptyset, V)$\\
                2&\qquad$Q'\leftarrow\textsc{greatestQuorum}(F, V\setminus Q)$\\
                3&\qquad{\bf if} $Q'\not=\emptyset$ {\bf return} $(Q,Q')$\\
                4&{\bf return $\mathit{true}$}\\\hline\\
                \hline \multicolumn{2}{l}{{\sc traverseMinQuorums}}\\\hline
                \multicolumn{2}{l}{\emph{Input:} FBAS $F$ with set of nodes $V$; $U\subseteq V$; $R\subseteq V\setminus U$}\\
                \multicolumn{2}{l}{\emph{Output:} enumerate all minimal quorums $Q$ of $F$ with}\\
                \multicolumn{2}{l}{\phantom{\emph{Output:}} $U\subseteq Q\subseteq U\cup R$ and $|Q|\leq |V|/2$}\\
                1&{\bf if} $|U|>|V|/2$ {\bf return}\\
                2&$Q=\textsc{greatestQuorum}(F, U)$\\
                3&{\bf if} $Q\not=\emptyset$\\
                4&\qquad{\bf if} $U=Q$ {\bf and not} $\textsc{containsProperSubQuorum}(F, U)$\\
                5&\qquad\qquad{\bf output} $U$\\
                6&{\bf else}\\
                7&\qquad{\bf if} $R\not=\emptyset$ {\bf and} $\textsc{greatestQuorum}(F, U\cup R, U)\not=\emptyset$\\
                8&\qquad\qquad$v\leftarrow \text{pick from $R$}$\\
                9&\qquad\qquad$\textsc{traverseMinQuorums}(F, U, R\setminus\{v\})$\\
                10&\qquad\qquad$\textsc{traverseMinQuorums}(F, U\cup\{v\}, R\setminus\{v\})$\\\hline
                \\
                \hline \multicolumn{2}{l}{{\sc containsProperSubQuorum}}\\\hline
                \multicolumn{2}{l}{\emph{Input:} FBAS $F$ with set of nodes $V$; $U\subseteq V$}\\
                \multicolumn{2}{l}{\emph{Output:} whether there is a quorum $Q\subsetneq U$}\\
                1&{\bf for} $v$ {\bf in} $U$\\
                2&\qquad{\bf if} $\textsc{greatestQuorum}(F, U\setminus \{v\})\not=\emptyset$ {\bf return} $\mathit{true}$\\
                3&{\bf return $\mathit{false}$}\\\hline
            \end{tabular}
        \end{center}
        \caption{Algorithm for {\sc quorumIntersection}.}
        \label{fig:alg-quorum-intersection}
    \end{figure}
    
    Now let us prove the correctness of {\sc traverseMinQuorums}; see Figure~\ref{fig:alg-quorum-intersection}. We remark that line~8 of {\sc traverseMinQuorums} nondeterministically picks an element from $R$. In practice we would use some heuristics to decide what element to pick next in order to minimize the run-time. These details shall not concern us in the present paper.
    \begin{lem}\label{lm:traverseMinQuorums}
       Given two sets $U\subseteq V$ and $R\subseteq V\setminus U$, the function invocation $\textsc{traverseMinQuorums}(F, U, R)$ enumerates every minimal quorum $Q$ of $F$ such that $U\subseteq Q\subseteq U\cup R$ and $|Q|\leq |V|/2$. Furthermore, every such quorum is enumerated exactly once.
    \end{lem}
    \begin{proof}
        We give a proof by induction over $|R|$. Before we proceed to the proof of the induction base and induction step, we will look at two special cases first.

        First consider the case $|U|>|V|/2$. Then there are no sets $Q$ with $U\subseteq Q\subseteq U\cup R$ and $|Q|\leq |V|/2$. Moreover, {\sc traverseMinQuorums} returns early in line~1 and, thus, does not enumerate any sets; i.e., the function behaves correctly.
        
        Next suppose that $|U|\leq |V|/2$ and that $U$ contains a quorum. In this case the function call will branch into line~4 because $Q$ is nonempty by Lemma~\ref{lm:alg-greatest-quorum}. Observe that the function {\sc containsProperSubQuorum} defined in Figure~\ref{fig:alg-quorum-intersection} correctly decides whether a set $U$ of nodes contains a proper subset that is a quorum. We conclude that the function {\sc traverseMinQuorums} will output $U$ in line~5 if and only if $U$ itself is a quorum (i.e., $U=Q$) and it is a minimal quorum. This is obviously correct behavior. The call to {\sc traverseMinQuorums} will then finish and not enumerate any other quorum. This is also correct because there cannot be any other minimal quorum $Q$ with $U\subseteq Q$.
        
        Let us now proceed to the actual inductive proof. It is sufficient to consider the case that $|U|\leq |V|/2$ and that $U$ does not contain a quorum -- the latter condition implies that $U$ itself is not a quorum. This case also implies that {\sc traverseMinQuorums} branches into line~7.
        
        For the base case $|R| = 0$ we have $R=\emptyset$ and {\sc traverseMinQuorums} will not branch into line~8, i.e., it does not enumerate any sets. This is correct behavior as the only set $Q$ satisfying $U\subseteq Q\subseteq U\cup R$ is $Q=U$ and it is not a quorum.
        
        For the induction step let $|R|>0$. Then $R\not=\emptyset$. We distinguish two cases. If $\textsc{greatestQuorum}(F, U\cup R, U)=\emptyset$, then (i)~there is no quorum $Q$ with $U\subseteq Q\subseteq U\cup R$ by Lemma~\ref{lm:alg-greatest-quorum} and (ii)~{\sc traverseMinQuorums} does not branch into line~8 and does not enumerate any quorums; hence {\sc traverseMinQuorums} behaves correctly. If $\textsc{greatestQuorum}(F, U\cup R, U)\not=\emptyset$, then {\sc traverseMinQuorums} branches into line 8, picks a node $v$ from $R$ and makes two recursive calls to {\sc traverseMinQuorums}. By induction, the first call enumerates all minimal quorums $Q$ with $U\subseteq Q\subseteq U\cup R$ that do not contain $v$ and the second call all those minimal quorums that do contain~$v$.
    \end{proof}

    \begin{ex}
        \label{ex:quorum-intersection-computation}
        In this experiment we study the computation time of the {\sc quorumIntersection} algorithm (see Figure~\ref{fig:alg-quorum-intersection}) implemented in \emph{Stellar Observatory} for the same sequence of FBAS as in Example~\ref{ex:quorum-enumeration-computation}. All FBAS of this form with $n>2$ have quorum intersection, and the computation time of the {\sc quorumIntersection} algorithm grows exponentially; see Figure~\ref{fig:quorum-intersection-computation}.
    \end{ex}

    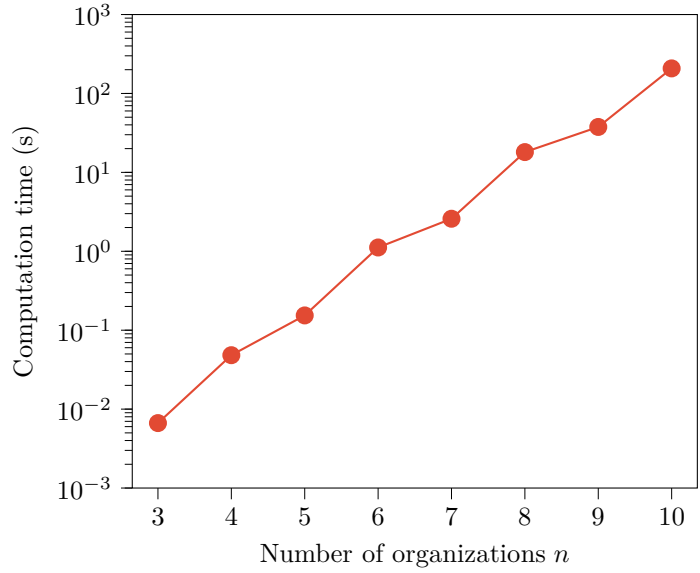
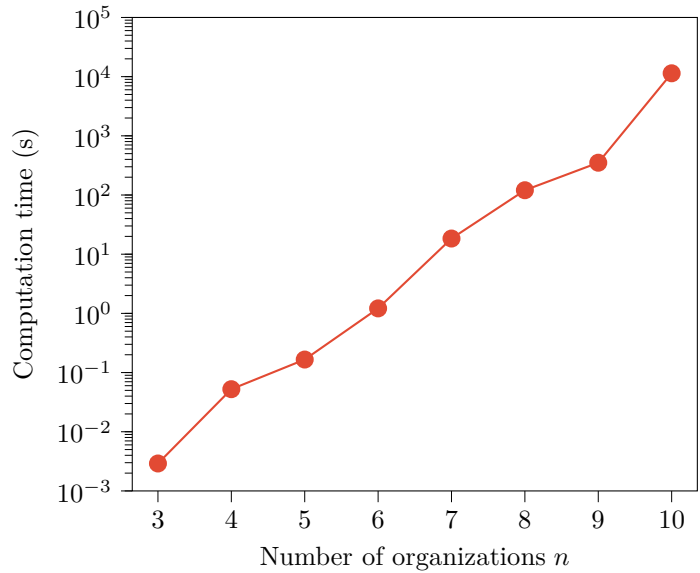
\begin{figure}[htbp]
        \centering
        \setlength{\figurewidth}{0.75\textwidth}
        \setlength{\figureheight}{0.65\textwidth}
        \begin{subfigure}[b]{\textwidth}
            \centering
\begin{tikzpicture}

\begin{axis}[
width=\figurewidth,
height=\figureheight,
log basis y={10},
tick align=outside,
tick pos=left,
x grid style={white!69.01960784313725!black},
xlabel={Number of organizations \(\displaystyle n\)},
xmin=2.65, xmax=10.35,
xtick style={color=black},
y grid style={white!69.01960784313725!black},
ylabel={Computation time (s)},
ymin=0.001, ymax=1000,
ymode=log,
ytick style={color=black}
]
\addplot [thick, plot_color_0, mark=*, mark size=3, mark options={solid}]
table {%
3 0.0066533088684082
4 0.0482308864593506
5 0.153916120529175
6 1.11756634712219
7 2.5849506855011
8 18.0781228542328
9 37.639146566391
10 207.770203113556
};
\label{fig:quorum-intersection-1:time}
\end{axis}

\end{tikzpicture}
            \caption{Root threshold $n-1$}
        \end{subfigure}
        \par\bigskip
        \begin{subfigure}[b]{\textwidth}
            \centering
\begin{tikzpicture}

\begin{axis}[
width=\figurewidth,
height=\figureheight,
log basis y={10},
tick align=outside,
tick pos=left,
x grid style={white!69.01960784313725!black},
xlabel={Number of organizations \(\displaystyle n\)},
xmin=2.65, xmax=10.35,
xtick style={color=black},
y grid style={white!69.01960784313725!black},
ylabel={Computation time (s)},
ymin=0.001, ymax=100000,
ymode=log,
ytick style={color=black}
]
\addplot [thick, plot_color_0, mark=*, mark size=3, mark options={solid}]
table {%
3 0.00290799140930176
4 0.0523691177368164
5 0.165960788726807
6 1.21394491195679
7 18.3959000110626
8 120.473617315292
9 351.319445133209
10 11395.9807598591
};
\end{axis}

\end{tikzpicture}
            \caption{Root threshold $\lfloor\frac{2}{3}n\rfloor+1$}
        \end{subfigure}
        
        \caption{Computation time of {\sc quorumIntersection} in seconds
        for the FBAS described in Example~\ref{ex:quorum-enumeration-computation} consisting of $n$ organizations with $3$ nodes each. Note that the y-axes are scaled logarithmically.}
        \label{fig:quorum-intersection-computation}
    \end{figure}

\section{Intactness}\label{sec:intact}
    A node $v\in V$ can be either \emph{well-behaved} or \emph{ill-behaved}\/\footnote{Alternative terms used in this context, which mean the same as ill-behaved, are \emph{faulty} (see, e.g.,~\cite[Section~3.1]{LokLosMaz19}) and \emph{Byzantine} (see, e.g.~\cite[Definition~1]{LosGafMaz19}).}.  
    This is a local property, which is independent of the quorum function $S$, and which other nodes cannot control. An important question is whether ill-behaved nodes can negatively impact other (well-behaved) nodes. If this happens, these other nodes are \emph{befouled} by the ill-behaved nodes. This intuitive idea will be made precise in Definitions~\ref{def:delete},~\ref{def:DSet} and~\ref{def:intact} below, 
    which correspond to the original development of Mazi\`eres~\cite[p.~9--10]{Maz16}:
    
    \begin{definition}[FBAS without $D$]\label{def:delete}
        Let $(V,S)$ be an FBAS and let $D\subseteq V$. The \emph{FBAS without $D$} is defined by
        $(V,S)^D=(V\setminus D,S^D)$, where 
        $$S^D(v)=\{s\setminus D\mid s\in S(v)\}~\text{for all}~v\in V\setminus D.$$
    \end{definition}
    
    Thus, forming $(V,S)^D$ amounts to deleting the nodes in $D$ from $V$ and from all quorum slices. The next result of Mazi\`eres will be useful in our discussion below.
    
    \begin{lem}[Theorem~1 in \cite{Maz16}]\label{lm:theo-1-maz16}
        Let $Q\subseteq V$ be a quorum in $(V,S)$ and let $D\subseteq V$. If $Q\setminus D\not=\emptyset$, then $Q\setminus D$ is a quorum in $(V,S)^D$.
    \end{lem}

    The following two corollaries are an immediate consequence of  Definition~\ref{def:delete}. The second one is a consequence of Lemma~\ref{lm:trust-cluster}(2).
    
    \begin{cor}\label{cor:delete}
        Let $D\subseteq V$ and $E\subseteq V\setminus D$. Then $((V,S)^D)^E=(V,S)^{D\cup E}$.
    \end{cor}
    
    \begin{cor}\label{cor:trust-cluster-deleted-fbas}
        Let $Z$ be a trust cluster of $(V,S)$. Then $(V,S)^{V\setminus Z}=(Z, S|_Z)$ and every quorum of $(V,S)^{V\setminus Z}$ is also a quorum of $(V,S)$.
    \end{cor}
    
    The following notion of dispensible sets (DSets) is an important stepping stone for the definition of intactness.
    
    \begin{definition}[DSet]\label{def:DSet}
        Let $(V,S)$ be an FBAS. A set $D\subseteq V$ is a \emph{dispensible set} or \emph{DSet}, when
        \begin{itemize}
            \item  $(V,S)^D$ has quorum intersection, and
            \item either $V\setminus D$ is a quorum in $(V,S)$ or $D=V$.
        \end{itemize}
    \end{definition}
    
    By this definition, $D=V$ is always a DSet, since $(V,S)^V=(\emptyset,S^D)$ with $S^D:\emptyset\rightarrow \{\{\emptyset\}\}$, consistently with Definition~\ref{def:FBAS}, and this FBAS trivially has quorum intersection. Whenever an FBAS $(V,S)$ has quorum intersection, $D=\emptyset$ is a DSet. An extreme case arises when $S(v)=\{V\}$ for all $v\in V$, since then $\emptyset$ and $V$ are the only DSets.
    
    \begin{ex}\label{ex:FBAS_fig7_Dsets}
        Consider the FBAS $(V,S)$ from Example~\ref{ex:FBAS_fig7}. Since this FBAS has quorum intersection, $\emptyset$ and $V$ are DSets. For every DSet $D\neq \emptyset$ it is necessary that $V\setminus D$ is a quorum in $(V,S)$. The only candidates are $\{1,2,3\}$, $\{4,5,6\}$, and $V\setminus\{7\}$. For each of these sets we see that the corresponding FBAS $(V,S)^D$ indeed has quorum intersection, and hence the DSets of $(V,S)$ are given by
        $$\emptyset, \quad\{1,2,3\},\quad \{4,5,6\},\quad 
        V\setminus\{7\},\quad V.$$ 
    \end{ex}
    
    We can now formalize the intuitive idea described at the beginning of this section. 
    
    \begin{definition}[intact and befouled]\label{def:intact}
        Let $B\subseteq V$ be the set of ill-behaved nodes in the FBAS $(V,S)$. A node $v$ is called \emph{$B$-intact} when there exists a \emph{single} DSet $D$ which contains \emph{all} ill-behaved nodes but not $v$, i.e., $B\subseteq D\subseteq V\setminus\{v\}$. If $B$ is obvious from the context, then we just say that $v$ is \emph{intact} instead of $B$-intact. If $v$ is not intact, it is called \emph{befouled}.
    \end{definition}
    
    Note that if a node $v$ is intact, then the definition immediately implies that it is well-behaved. Conversely, if $v$ is ill-behaved, then it is befouled. The interesting object of study are the well-behaved nodes that are befouled.
    
     We have shown in Theorem~\ref{thm:quorumsplit-np} that the problem \textsc{SimpleQuorumSplit} is NP-complete. Since the problem to decide whether a given set $D\subseteq V$ is a DSet involves the problem to decide whether $(V,S)^D$ has quorum intersection, the DSet decision problem is computationally at least as hard. We will deal with the algorithmic treatment of intactness in Section~\ref{sec:algo-intact} below.

\subsection{Intact nodes in symmetric simple FBAS}
    While the computation of quorums and DSets in general is computationally challenging, the quorums of a symmetric simple FBAS can be completely characterized, see Lemma~\ref{lem:simplequorum}. The following result shows that a simple characterization also exists for DSets of symmetric simple FBAS. 
    
    \begin{lem}\label{lem:simpledelete}
        Let $(V,k)$ be a symmetric simple FBAS, and let $D\subseteq V$. Then the set of all quorums of $(V,k)^D$ is given by $\bigl\{Q\subseteq V\setminus D\,\mid\,|Q|\geq \max\{1,k-|D|\}\bigr\}$.
        Moreover, $(V,k)^D$ has quorum intersection if and only if $\max\{1,k-|D|\}>(|V|-|D|)/2$.
    \end{lem}
    \begin{proof}
        If $Q\subseteq V\setminus D$ is any set with $|Q|<\max\{1,k-|D|\}$ then either $Q$ is empty and hence not a quorum, or $1\leq |Q|<k-|D|$. Let $v\in Q$, and let $s$ be any quorum slice of $v$ in $(V,k)$. From $|s|=k$ we obtain $|s\setminus D|\geq |s|-|D|=k-|D|>|Q|$, which shows that $s\setminus D$ is not contained in $Q$, and hence $Q$ cannot be a quorum.
        
        Now let $Q\subseteq V\setminus D$ with $|Q|\geq \max\{1,k-|D|\}$; then $Q$ is nonempty. Consider a node $v\in Q$.
        
        (i) If $|D|>k-1$, then there exist distinct $d_1,\dots,d_{k-1}\in D$. Then $\{v\}\cup \{d_1,\dots,d_{k-1}\}$ is a quorum slice of $v$ in $(V,k)$ and, thus, $\{v\}$ is a quorum slice of $v$ in $(V,k)^D$ that is contained in $Q$. 
        
        (ii) If $|D|\leq k-1$, then $|Q|\geq k-|D|$ implies that there exist distinct elements $q_1,\dots,q_{k-|D|-1}\in Q\setminus\{v\}$. Then $\{v\}\cup\{q_1,\dots,q_{k-|D|-1}\}\cup D$ 
        is a quorum slice of $v$ in $(V,k)$, and thus $\{v\}\cup\{q_1,\dots,q_{k-|D|-1}\}$ is a quorum slice of $v$ in $(V,k)^D$ that is contained in $Q$.
        
        This concludes the proof that $Q$ is a quorum of $(V,k)^D$. The statement that $(V,k)^D$ has quorum intersection if and only if $\max\{1,k-|D|\}>(|V|-|D|)/2$ can be proved similarly to Lemma~\ref{lem:simplequorum}.
    \end{proof}
    
    We remark that deleting nodes from a symmetric simple FBAS according to Definition~\ref{def:delete} does not preserve the property to be symmetric and simple. To show this consider the example symmetric simple FBAS $(V,k)$ with $V=\{a,b,c\}$ and $k=2$. Let $D=\{c\}$. Then the quorum slices of $a$ in $(V,k)$ are $\{a,b\}$ and $\{a,c\}$ and the quorum slices of $a$ in $(V,k)^D$ are $\{a, b\}$ and $\{a\}$. These last two quorum slices have different sizes so that we can conclude that $(V,k)^D$ is not simple and therefore not symmetric.
    
    Using the Lemmas~\ref{lem:simpledelete} and~\ref{lem:simplequorum} we now characterize the DSets of a simple symmetric FBAS.
    
    \begin{lem}\label{lem:symmDSet}
        Let $(V,k)$ be a symmetric simple FBAS. A set $D\subseteq V$ is a DSet if and only if one of the following holds:
        \begin{enumerate}[(1)]
            \item $D=V$, or
            \item $k=1$ and $|D|=|V|-1$, or
            \item $|D|\leq\min\{|V|-k,\,2k-|V|-1\}$.
        \end{enumerate}
    \end{lem}
    \begin{proof}
        By definition, $D$ is a DSet if and only if 
        \begin{enumerate}[(i)]
            \item $D=V$ or
            \item $\max\{1,k-|D|\}>(|V|-|D|)/2$, i.e., $(V,k)^D$ has quorum intersection (cf. Lemma~\ref{lem:simpledelete}) and \\$|V|-|D|\geq k$, i.e., $V\setminus D$ is a quorum in $(V,k)$ (cf. Lemma~\ref{lem:simplequorum}).
        \end{enumerate}
        The equivalence is now shown by the following implications:
        
        ``(ii) $\Rightarrow$ (1), (2) or (3)'': If $\max\{1,k-|D|\}>(|V|-|D|)/2$, then either (a) $1>(|V|-|D|)/2$, or (b) $k-|D|>(|V|-|D|)/2$. In case (a) we have either $V=D$ and hence statement (1), or $|V|-|D|=1$, which together with $|V|-|D|\geq k$ yields (2). In case (b) we have $|D|<2k-|V|$, which together with $|D|\leq |V|-k$ yields (3).
        
        ``(2) or (3) $\Rightarrow$ (ii)'':  Clearly, (2) implies (ii). Moreover, (3) yields $|D|\leq |V|-k$ and hence $|V|-|D|\geq k$, as well as $|D|\leq 2k-|V|-1$ or equivalently $k-|D|>(|V|-|D|)/2$, so that we again obtain (ii).
    \end{proof}
    
    Lemma~\ref{lem:symmDSet} shows that if $(V,k)$ is a symmetric simple FBAS with $1<k=|V|$ or $1<k\leq (|V|+1)/2$, then $(V,k)$ has only the trivial DSets $D=V$ and possibly $D=\emptyset$.  
    
    We say that an FBAS $(V,S)$ is \emph{resilient} against $m\geq 0$ ill-behaved nodes if no well-behaved node becomes befouled whenever there are at most $m$ ill-behaved nodes. We now ask about the minimum size of $V$ such that for a suitable $k$ there is a symmetric simple FBAS $(V,k)$ that is resilient against $m$ ill-behaved nodes. We only want to consider FBAS that have quorum intersection and therefore we require $k>|V|/2$ by Lemma~\ref{lem:simplequorum}. In particular, $k>1$ if we exclude trivial FBAS.
    
    We know that a node $v$ is intact when all ill-behaved nodes (at most $m$) are contained in a single DSet that does not contain $v$. By Lemma~\ref{lem:symmDSet}, the relevant DSets of $(V,k)$ are given by \emph{all} sets $D\subseteq V\setminus\{v\}$ which satisfy $m\leq |D|\leq \min\{|V|-k, 2k-|V|-1\}$, and hence $|V|$ and $k$ must satisfy $m\leq \min\{|V|-k, 2k-|V|-1\}$. Maximizing the right hand side over $k$ by equating the two expressions yields $k=2|V|/3+1/3$, and then $m\leq |V|-k$ leads to $|V|\geq 3m+1$. Thus, the minimum number of nodes is $|V|=3m+1$, and the corresponding constant is $k=2(3m+1)/3+1/3=2m+1$. 
    
    We thus have shown the following result, which is consistent with traditional Byzantine fault tolerance; see, e.g.,~\cite{CasLis02},~\cite[Theorem~4]{BraTou85}, and Example~\ref{ex:Byzantine} above.
    
    \begin{thm}\label{thm:SymmetricByzantine}
        The minimum number of nodes of a symmetric simple FBAS $(V,k)$ that is resilient against (at most) $m$ ill-behaved nodes is $3m+1$, and this FBAS has constant $k=2m+1$.
    \end{thm}

\subsection{Comparison of intactness definitions}\label{sec:intact-compare}
    In the recent publications~\cite{LokLosMaz19,LosGafMaz19}, quorums and intactness are defined differently from~\cite{Maz16}, and in this subsection we will compare both variants. Let us first give an intuition for the main differences:
    \begin{itemize}
        \item In the definitions of~\cite{Maz16}, which we adopted in this paper, the main notions are introduced in the following order: First, FBAS, quorums, quorum intersection, and DSets are defined without any reference to ill-behaved nodes. Then one postulates a set of ill-behaved nodes and introduces the concept of an intact node, which depends on the set of ill-behaved nodes.
        \item In~\cite{LokLosMaz19,LosGafMaz19}, the set of intact nodes is directly ``baked into'' an FBAS, i.e., in these papers FBAS are richer structures that already define their intact nodes. Then quorums and quorum intersection directly incorporate the set of intact nodes into their definition. Finally, the notion of \emph{intact sets} is introduced.
    \end{itemize}
    The advantage of the first paradigm is that FBAS can be viewed as structures agnostic to the set of ill-behaved nodes (which is usually unknown in practice). An advantage of the second paradigm is its simplicity -- it does not require definitions of DSets or FBAS of the form $(V,S)^D$.
    
    In order to mathematically compare the two paradigms, we need the following definition.
    \begin{definition}[$B$-quorum and $B$-intact set] \label{def:B-intact-set}
        Let $(V,S)$ be an FBAS and $B\subseteq V$. A \emph{$B$-quorum} is a nonempty set $Q\subseteq V$ such that for every $v\in Q\setminus B$ there is a quorum slice $s\in S(v)$ with $s\subseteq Q$. A subset $U\subseteq V\setminus B$ is called a \emph{$B$-intact set} if 
        \begin{itemize}
            \item $U$ is a $B$-quorum or $U=\emptyset$,
            \item for all $B$-quorums $Q$, $Q'$ that intersect $U$, also $Q\cap Q'$ intersects $U$.
        \end{itemize}
    \end{definition}
    Note that in Definition~\ref{def:intact} we have defined $B$-intact \emph{nodes}, whereas Definition~\ref{def:B-intact-set} introduces $B$-intact \emph{sets} (as in~\cite[Section~3.1]{LokLosMaz19} and~\cite[Section~5]{LosGafMaz19}). 
    Intuitively, we expect that a node is $B$-intact if and only if it is contained in a $B$-intact set. We will show that this is actually not the case, and that the former property is strictly stronger than the latter.
    
    For the following discussion we fix a set $B\subseteq V$, the set of ill-behaved nodes.
    \begin{lem}\label{lm:intactness-comp-inequal2}
        If a node $v$ is $B$-intact, then it is contained in a $B$-intact set.
    \end{lem}
    \begin{proof}
        Let us first prove the following auxiliary statement:
        \begin{itemize}[(A)]
            \item Let $D\subseteq V$ contain $B$ and let $Q$ be a $B$-quorum with $Q\not\subseteq D$. Then $Q\setminus D$ is a quorum in $(V,S)^D$.
        \end{itemize}
        Obviously, $Q\setminus D$ is nonempty. Let $v\in Q\setminus D$. Then in particular $v\not\in B$. Hence, there is an $s\in S(v)$ with $s\subseteq Q$ and, therefore, $s\setminus D\subseteq Q\setminus D$. We conclude that $Q\setminus D$ is a quorum in $(V,S)^D$. This concludes the proof of Statement~(A).
    
        Now let us proceed with the proof of this lemma. Let $v$ be $B$-intact. Then there is DSet $D$ containing $B$ but not $v$. Let $U=V\setminus D$. We will show that $U$ is a $B$-intact set. Clearly, $U$ is a quorum because $U$ is nonempty and due to the definition of $D$. Then $U$ is also a $B$-quorum. Let $Q,Q'$ be $B$-quorums that intersect $U$. By Statement~(A), $Q\setminus D$ and $Q'\setminus D$ are quorums in $(V, S)^D$. Since $D$ is a DSet we obtain that $\emptyset\not=(Q\setminus D)\cap (Q'\setminus D)=(Q\cap Q')\setminus D$, i.e., $Q\cap Q'$ intersect $U$.
    \end{proof}
    
    The following example shows that the converse of Lemma~\ref{lm:intactness-comp-inequal2} does not hold in general, i.e., we cannot conclude that every member of a $B$-intact set is $B$-intact itself.
    \begin{ex}
        Let $V=\{a,b,c,d\}$, 
        \begin{align*}
            S(a)&=\{\{a,b\}\},&
            S(b)&=\{\{a,b,c,d\}\},\\
            S(c)&=\{\{b,c\},\{c,d\}\},&
            S(d)&=\{\{b,d\}, \{c,d\}\},
        \end{align*}
        and let $B=\{a\}$. Then there are two quorums, $\{c,d\}$ and $\{a,b,c,d\}$, and four $B$-quorums, $\{a\}$, $\{a,c,d\}$, $\{c,d\}$ and $\{a,b,c,d\}$, so that $\{c,d\}$ is a $B$-intact set.\\
        On the other hand, the only DSet that contains $B$ is $V$ itself, and hence 
        no node in $V$ is $B$-intact. (Observe that $\{a,b\}$ is not a DSet as $(V, S)^{\{a,b\}}$ has the quorums $\{c\}$ and $\{d\}$ and therefore has no quorum intersection.)
    \end{ex}
    
    We will now introduce an FBAS property under which the converse of  Lemma~\ref{lm:intactness-comp-inequal2} is guaranteed to hold. 
    \begin{definition}[subslice property]
        We say that $(V, S)$ has the \emph{$B$-subslice property} if for every $v\in V\setminus B$, $s_v\in S(v)$ and $u\in s_v\setminus B$ there is some $s_u\in S(u)$ with $s_u\subseteq s_v$.
    \end{definition}
    We remark that the $B$-subslice property reminds of the main property of Personal Byzantine Quorum Systems \cite[Property~1]{LosGafMaz19}. It is easy to see that the $B$-subslice property holds for symmetric simple FBAS for every set $B$. Moreover, this property holds for many example FBAS presented in this paper such as Examples~\ref{ex:trivial}, \ref{ex:two}, \ref{ex:FBAS_fig7} and~\ref{ex:hierarchy}. More interestingly, it is easy to show that also the greatest SCC of the Stellar network (see Example~\ref{ex:stellar-core}) has the $B$-subslice property for every set $B$ of ill-behaved nodes.
    
    \begin{thm}\label{thm:intactness-comp-equal}
        Let $(V,S)$ have the $B$-subslice property and let $v\in V$. Then the following statements are equivalent:
        \begin{enumerate}[(1)]
            \item $v$ is $B$-intact.
            \item $v$ is contained in a $B$-intact set.
        \end{enumerate}
    \end{thm}
    \begin{proof}
        The direction (1) $\Rightarrow$ (2) follows from Lemma~\ref{lm:intactness-comp-inequal2}.
        
        Let us now prove the direction (2) $\Rightarrow$ (1) and suppose that $v$ is contained in a $B$-intact set $U$. It suffices to show that $D=V\setminus U$ is a DSet because $v\in U$ and $U\cap B=\emptyset$. By the last property, if $U$ is a $B$-quorum, then it is also a quorum. Hence, the only non-trivial property that we need to show is that $(V,S)^D$ has quorum intersection. Let $Q, Q'$ be quorums in $(V,S)^D$. We will show that there are $B$-quorums $P$, $P'$ with $P\setminus D= Q$ and $P'\setminus D=Q'$. We only prove that there is such a set $P$; the proof for the existence of $P'$ follows by symmetry.
        
        For every $v\in Q$ there is a slice $s_v\in S(v)$ such that $s_v\setminus D\subseteq Q$. Let $P=\bigcup_{v\in Q}s_v$. We will show that (i)~$P$ is a $B$-quorum and that (ii)~$P\setminus D=Q$.
        
        (i)~Let $u\in P\setminus B$. Then there is some $v\in Q$ with $u\in s_v$. Since $v\in Q\subseteq V\setminus D$, we have $v \not\in B$ and $u\not\in B$. Then there is some $s_u\in S(u)$ with $s_u\subseteq s_v\subseteq P$ by the $B$-subslice property.
        
        (ii)~$P\setminus D = \bigcup_{v\in Q}s_v\setminus D=Q$ because $v\in s_v\setminus D\subseteq Q$ for every $v\in Q$.
        
        This concludes the proof for the existence of $P$ and $P'$. Since $Q$ and $Q'$ are contained in $U$ and nonempty, the sets $P$ and $P'$ intersect $U$. Thus, also $P\cap P'$ intersects $U$ because $U$ is a $B$-intact set. Thus, $\emptyset\not=(P\cap P')\cap U=(P\cap P')\setminus D=(P\setminus D)\cap (P'\setminus D)=Q\cap Q'$. We conclude that $(V,S)^D$ has quorum intersection.
    \end{proof}

\section{Algorithmic treatment of intactness}\label{sec:algo-intact}
    Let us now study the problem to algorithmically determine the set of $B$-intact nodes given a set $B$ of ill-behaved nodes. We point out that in a real-life situation one does usually not know the set $B$ and would not be able to apply such an algorithm. However, it is useful for analyzing a given FBAS and to simulate certain failure and attack scenarios.
    
    The techniques we use in this section are based on methods and algorithms developed in Section~\ref{sec:alg-quorums}. Throughout this section we will assume that we only consider FBAS $(V,S)$ that have quorum intersection.

\subsection{Strongly connected components}\label{sec:algo-intact-scc}
    In this section we revisit the trust graph of the FBAS and show how it affects the intactness of nodes in an FBAS. We will see that analyzing the structure of the SCCs of the trust graph allows us to prune major parts of the FBAS when determining intactness of nodes.
    
    Let us first present two rather technical properties about trust clusters and DSets. The first one ensures that every DSet still behaves like a DSet within a trust cluster. The second one states that every DSet within a trust cluster can be extended to a DSet for the whole FBAS.
    
    \begin{lem}\label{lm:SCC-intact-down}
        Let $Z$ be a trust cluster of $(V,S)$ and let $D$ be a DSet of $(V,S)$. Then $D\cap Z$ is a DSet of $(Z, S|_Z)$.
    \end{lem}
    \begin{proof}
        By Corollary~\ref{cor:trust-cluster-deleted-fbas} $(Z, S|_Z)=(V, S)^{V\setminus Z}$. Thus, we need to show that
        \begin{enumerate}[(1)]
            \item $((V, S)^{V\setminus Z})^{D\cap Z}$ has quorum intersection.
            \item $Z\setminus (D\cap Z)$ is a quorum in $(V, S)^{V\setminus Z}$ or $D\cap Z=Z$. 
        \end{enumerate}
        (1)~The equality $(V\setminus Z)\cup(D\cap Z)=D\cup((V\setminus D)\setminus (Z\setminus D))$ and Corollary~\ref{cor:delete} imply $((V, S)^{V\setminus Z})^{D\cap Z}=((V, S)^D)^{(V\setminus D)\setminus (Z\setminus D)}$. Let us refer to this FBAS as $F'$. It is straightforward to show that $Z\setminus D$ is a trust cluster of $(V, S)^D$. Then Corollary~\ref{cor:trust-cluster-deleted-fbas} implies that every quorum of $F'$ is also a quorum of $(V,S)^D$. Thus, $F'$ has quorum intersection because $(V,S)^D$ does have quorum intersection.
        
        (2)~We will show the equivalent statement that $Z\setminus D$ is a quorum in $(V, S)^{V\setminus Z}$ or $Z\subseteq D$. Since $D$ is a DSet we have either that $V\setminus D$ is a quorum in $(V,S)$ or that $D=V$. We are done in the latter case as it implies $Z\subseteq D$. Consider the former case and suppose that $Z\not\subseteq D$. We will show that $Z\setminus D$ is a quorum in $(V, S)^{V\setminus Z}$. From $Z\not\subseteq D$ we obtain $(V\setminus D)\setminus (V\setminus Z)\not=\emptyset$. Then Lemma~\ref{lm:theo-1-maz16} yields that $(V\setminus D)\setminus (V\setminus Z)=Z\setminus D$ is a quorum in $(V,S)^{V\setminus Z}$.
    \end{proof}
    
    \begin{lem}\label{lm:SCC-intact-up}
        Let $Z$ be a trust cluster of $(V,S)$ and let $D$ be a DSet of $(Z, S|_Z)$. Then $D\cup (V\setminus Z)$ is a DSet of $(V, S)$.
    \end{lem}
    \begin{proof}
        Our assumption is equivalent to the statement that $D$ is a DSet of $(V, S)^{V\setminus Z}$ due to Corollary~\ref{cor:trust-cluster-deleted-fbas}. We need to show that
        \begin{enumerate}[(1)]
            \item $(V, S)^{D\cup (V\setminus Z)}$ has quorum intersection.
            \item $V\setminus (D\cup (V\setminus Z))$ is a quorum in $(V, S)$ or $D\cup (V\setminus Z)=V$. 
        \end{enumerate}
        Statement~(1)~follows from Corollary~\ref{cor:delete} and the fact that $((V,S)^{V\setminus Z})^D$ has quorum intersection, which is a consequence of our assumption. Next let us prove Statement~(2). We have that $Z\setminus D$ is a quorum in $(V, S)^{V\setminus Z}$ or $D=Z$. In the former case Corollary~\ref{cor:trust-cluster-deleted-fbas} yields that $V\setminus (D\cup (V\setminus Z))=Z\setminus D$ is a quorum in $(V,S)$. The latter case clearly implies $D\cup (V\setminus Z)=V$.
    \end{proof}
    We can now combine the previous two lemmas to the following result.
    \begin{thm}\label{theo:intact-SCC}
        Let $Z$ be a trust cluster, $v\in Z$ and $B\subseteq V$.
        The node $v$ is $B$-intact in $(V,S)$ if and only if $v$ is $(B\cap Z)$-intact in $(Z,S|_Z)$.
    \end{thm}
    \begin{proof}
        ``$\Rightarrow$'': Let $v$ be $B$-intact in $(V,S)$. Then there is a DSet $D$ of $(V,S)$ with $B\subseteq D\subseteq V\setminus\{v\}$. Lemma~\ref{lm:SCC-intact-down} yields that $D\cap Z$ is a DSet of $(Z, S|_Z)$. Clearly, $B\cap Z\subseteq D\cap Z\subseteq Z\setminus\{v\}$, which proves that $v$ is $(B\cap Z)$-intact in $(Z,S|_Z)$.
        
        ``$\Leftarrow$'': Let $v$ be $(B\cap Z)$-intact in $(Z,S|_Z)$. Then there is a DSet $D$ of $(Z, S|_Z)$ with $B\cap Z\subseteq D\subseteq Z\setminus \{v\}$. Lemma~\ref{lm:SCC-intact-up} implies that $D\cup(V\setminus Z)$ is a DSet of $(V,S)$. Note that $B\subseteq(B\cap Z)\cup(V\setminus Z)\subseteq D\cup(V\setminus Z)\subseteq (Z\setminus \{v\})\cup(V\setminus Z)=V\setminus\{v\}$ and therefore, $v$ is $B$-intact in $(V,S)$.
    \end{proof}
    The following corollary uses Lemma~\ref{lem:trust-cluster}.
    \begin{cor}\label{lm:intact-SCC}
        Let $C$ be the greatest SCC and $B\subseteq V$. Let $I$ be the set of $B$-intact nodes of $(V,S)$. Then $I\cap C$ is the set of $(B\cap C)$-intact nodes of $(C,S|_C)$.
    \end{cor}
    
    There are two ways to interpret the previous theorem and corollary. First, whether nodes inside a trust cluster are intact or befouled merely depends on the ill-behaved nodes inside the trust cluster. The ill-behaved nodes outside the trust cluster have no influence on this. Second, if we are only interested to know what nodes inside a specific trust cluster $Z$ are intact, we can apply the algorithm that we present below to the smaller FBAS $(Z, S|_Z)$ instead. This is particularly meaningful when one is only interested in determining what nodes in the greatest SCC are intact. It also shows that the nodes in the greatest SCC have the highest importance in an FBAS -- they can influence the intactness of all other nodes but their own intactness cannot be influenced by nodes outside the SCC.

\subsection{The intactness algorithm}
    The algorithm to compute the set of $B$-intact nodes is shown in Figure~\ref{fig:alg-intactness}. Before we are going to prove its correctness, we want to point out that in a naive implementation of the expression $\textsc{quorumIntersection}(F^{V\setminus Q})$ on line~4 one would need to construct $F^{V\setminus Q}$ first. This can be challenging as usually $F^{V\setminus Q}$ is not a simple FBAS even if $F$ is a simple FBAS; see the remark after Lemma~\ref{lem:simpledelete}. Then the construction of $F^{V\setminus Q}$ would lead to an exponential blowup; cf. our concept of size of general and simple FBAS in Definition~\ref{def:fbas-size}.
    
    There is a better approach that avoids that problem. We can redefine the function $\textsc{quorumIntersection}$ to accept an additional parameter $D$, so that the function call $\textsc{quorumIntersection}(F, D)$ determines whether $F^D$ has quorum intersection. This can be achieved by redefining the functions {\sc traverseMinQuorums}, {\sc greatestQuorum} and {\sc containSlice}, too, and equip them with an additional parameter $D$ in a similar fashion. Then the parameter $D$ is forwarded from {\sc quorumIntersection} to {\sc greatestQuorum} to {\sc containsSlice} and every reference to $V$ in these algorithms is replaced by $V\setminus D$. Additionally we need to change the logic of the function {\sc containsSlice}: replace line~2 for the general FBAS version with
    \begin{center}
        {\bf if} $s\setminus D\subseteq U$ {\bf return} $\mathit{true}$
    \end{center}
    and replace line~1 for the simple FBAS version with
    \begin{center}
        {\bf return} $|(U\cup D)\cap q(v)|\geq n(v)$
    \end{center}
    While the change to the general FBAS version is obvious, the change to the simple FBAS version requires some argument.
    
    \begin{lem}
        Let $F=(V, q, n)$ be a simple FBAS, let $D, U\subseteq V$ and $v\in U\setminus D$. Then $v$ has a quorum slice in $F^D$ contained in $U$ if and only if $|(U\cup D)\cap q(v)|\geq n(v)$.
    \end{lem}
    \begin{proof}
        First we expand the simple FBAS $(V, q, n)$ into its general form $(V, S)$, i.e., $S(v)=\{W\subseteq q(v)\mid v\in W, |W|=n(v)\}$. Next we construct $F^D=(V\setminus D, S^D)$ with
        \begin{align*}
            S^D(v)&=\{s\setminus D\mid s\in S(v)\}
            =\{W\setminus D\mid W\subseteq q(v), v\in W, |W|=n(v)\}.
        \end{align*}
        Then the following statements are equivalent: (i)~$v$ has a quorum slice in $F^D$ contained in $U$, and (ii)~there is a set $W\subseteq q(v)$ with $v\in W$, $|W|=n(v)$ and $W\setminus D\subseteq U$. It remains to prove that the latter statement is equivalent to  (iii)~$|(U\cup D)\cap q(v)|\geq n(v)$. For the direction (ii)~$\Rightarrow$~(iii) we get $W\subseteq U\cup D$ and therefore $W\subseteq (U\cup D)\cap q(v)$. Hence, $n(v)=|W|\leq |(U\cup D)\cap q(v)|$. For the direction (iii)~$\Rightarrow$~(ii) we have $v\in (U\cup D)\cap q(v)$, and therefore there is a subset $W\subseteq (U\cup D)\cap q(v)$ with $|W|=n(v)$ and $v\in W$. This set $W$ clearly satisfies~(ii).
    \end{proof}
    
    \begin{figure}[tb]
        \begin{center}
            \begin{tabular}{rl}
                \hline \multicolumn{2}{l}{{\sc intactNodes}}\\\hline
                \multicolumn{2}{l}{\emph{Input:} FBAS $F$ having quorum intersection with set of nodes $V$; $B\subseteq V$}\\
                \multicolumn{2}{l}{\emph{Output:} the set of all $B$-intact nodes}\\
                1&$U_1\rightarrow V\setminus B$, $i\leftarrow 1$\\
                2&{\bf repeat}\\
                3&\qquad$Q\leftarrow \textsc{greatestQuorum}(F, U_i)$\\
                4&\qquad$\mathrm{result} \leftarrow \textsc{quorumIntersection}(F^{V\setminus Q})$\\
                5&\qquad{\bf if} $\mathrm{result}=\mathit{true}$\\
                6&\qquad\qquad{\bf return} $Q$\\
                7&\qquad{\bf else}\\
                8&\qquad\qquad$(Q_1, Q_2)\leftarrow \mathrm{result}$\\
                9&\qquad\qquad$W_1\leftarrow \textsc{greatestQuorum}(F, Q\setminus Q_1)$\\
                10&\qquad\qquad$W_2\leftarrow \textsc{greatestQuorum}(F, Q\setminus Q_2)$\\
                11&\qquad\qquad{\bf if} $W_1=\emptyset$\\
                12&\qquad\qquad\qquad$U_{i+1}\leftarrow W_2$\\
                13&\qquad\qquad{\bf else if} $W_2=\emptyset$\\
                14&\qquad\qquad\qquad$U_{i+1}\leftarrow W_1$\\
                15&\qquad\qquad{\bf else}\\
                16&\qquad\qquad\qquad$U_{i+1}\leftarrow W_1\cap W_2$\\
                17&\qquad$i\leftarrow i+1$\\\hline
            \end{tabular}
        \end{center}
        \caption{Algorithm for determining $B$-intact nodes given a set $B$ of nodes. The functions {\sc greatestQuorum} and {\sc quorumIntersection} are defined in Section~\ref{sec:alg-quorums}.}        \label{fig:alg-intactness}
    \end{figure}
    
    Let us now prove the correctness of the algorithm {\sc intactNodes} in Figure~\ref{fig:alg-intactness}. First we state an auxiliary lemma.
    \begin{lem}\label{lm:alg-intactness-aux}
        Consider the $i$-th iteration of the loop of {\sc intactNodes} and the sets $Q$ and $U_{i+1}$ constructed in that iteration. Suppose that there is a DSet $D$ with $V\setminus D\subseteq U_i$. Then $V\setminus D\subseteq Q$ and if the iteration branches into line~8, also $V\setminus D\subseteq U_{i+1}$.
    \end{lem}
    \begin{proof}
        This statement is trivial if $D=V$. So let us assume that $D\not=V$; then $V\setminus D$ is a quorum because $D$ is a DSet. Line~3 yields that $Q$ is a quorum with $V\setminus D\subseteq Q$.
        
        Now suppose that the iteration branches into line~8. Let $Q_1$, $Q_2$, $W_1$ and $W_2$ be the values of the respective variables in iteration $i$. Let $a\in\{1,2\}$. We will show that
        \begin{enumerate}
            \item if $Q_a\subseteq D$, then $V\setminus D\subseteq W_a$,
            \item if $Q_a\not\subseteq D$, then $Q_a\setminus D$ is a quorum in $F^D$ and $W_a=\emptyset$,
            \item $Q_1\subseteq D$ or $Q_2\subseteq D$.
        \end{enumerate}
        For the proof of the first claim we have that $V\setminus D$ is a quorum contained in $Q\setminus Q_a$, and hence $V\setminus D\subseteq W_a$.
        
        For the proof of the second claim observe that $Q_a\setminus D\not=\emptyset$. In particular, $Q_a\setminus (D\cap Q)\not=\emptyset$. 
        Together with Lemma~\ref{lm:theo-1-maz16} applied to $F^{V\setminus Q}$ and the fact that $Q_a$ is a quorum of $F^{V\setminus Q}$, we obtain that $Q_a\setminus (D\cap Q)$ is a quorum in $(F^{V\setminus Q})^{D\cap Q}$. Clearly, $Q_a\setminus (D\cap Q)=Q_a\setminus D$ because $Q_a\subseteq Q$ by construction in line~4. Furthermore, $(F^{V\setminus Q})^{D\cap Q}=F^D$ by Corollary~\ref{cor:delete} and because $(V\setminus Q)\cup(D\cap Q)=D$ due to $V\setminus D\subseteq Q$. We conclude that $Q_a\setminus D$ is a quorum in $F^D$. 
        
        For the second claim it remains to show that $W_a=\emptyset$. We assume that $W_a$ is nonempty and derive a contradiction. Then $W_a$ is a quorum by construction and it intersects $V\setminus D$ because we assumed that $F$ has quorum intersection; equivalently, $W_a\setminus D\not=\emptyset$. By Lemma~\ref{lm:theo-1-maz16} applied to $F$ we obtain that $W_a\setminus D$ is a quorum in $F^D$. The fact that $D$ is a DSet implies that $F^D$ has quorum intersection and, therefore, $Q_a\setminus D$ and $W_a\setminus D$ intersect. This is a contradiction to $W_a\subseteq Q\setminus Q_a$, which follows from the construction of $W_a$.
        
        For the proof of the third claim assume that $Q_1\not\subseteq D$ and that $Q_2\not\subseteq D$. We will derive a contradiction. By the second claim, $Q_1\setminus D$ and $Q_2\setminus D$ are quorums in $F^D$. They intersect because $D$ is a DSet and, thus, $F^D$ has quorum intersection. This 
        contradicts that $Q_1$ and $Q_2$ are disjoint by construction.
        
        Now we can employ our three claims to prove the lemma. By the third claim $Q_1\subseteq D$ or $Q_2\subseteq D$. Due to symmetry we will just consider the case that $Q_1\subseteq D$. Then $V\setminus D\subseteq W_1$ by the first claim and, hence, $W_1\not=\emptyset$ because we assumed $D\not=V$. Therefore the algorithm will not branch into line~12. If $Q_2\not\subseteq D$, then by the second claim the algorithm will branch into line~14 and we have $V\setminus D\subseteq W_1=U_{i+1}$. If $Q_2\subseteq D$, then again $\emptyset\not=V\setminus D\subseteq W_2$ and the algorithm will branch into line~16 and we obtain $V\setminus D\subseteq W_1\cap W_2=U_{i+1}$.
    \end{proof}
    
    \begin{thm}
        Given a set $B$ the function {\sc intactNodes} returns the set of all $B$-intact nodes.
    \end{thm}
    \begin{proof}
        First we show that the function terminates by proving that the sequence $U_1, U_2, \ldots$ is strictly decreasing, i.e., $U_{i+1}\subsetneq U_i$ for every $i$ as long as the loop does not terminate in line~6. Clearly, (i)~$Q\subseteq U_i$ (line~3), (ii)~$Q_1\not=\emptyset\not=Q_2$ as they are quorums, and (iii) $W_1\subseteq Q\setminus Q_1$, $W_2\subseteq Q\setminus Q_2$. These statements together imply $W_1\subsetneq U_i$ and $W_2\subsetneq U_i$. Then $U_{i+1}\subsetneq U_i$ because $U_{i+1}\subseteq W_1$ or $U_{i+1}\subseteq W_2$.
        
        Next observe that when {\sc intactNodes} terminates and returns a set $Q$ in line~6, then $V\setminus Q$ is a DSet because $F^{V\setminus Q}$ has quorum intersetion (line~4) and because $Q$ is a quorum or empty (line~3). Moreover, $V\setminus Q$ contains $B$ because $U_1, U_2, \ldots$ is a decreasing sequence and $U_1\cap B=\emptyset$ by construction and therefore also $Q\cap B=\emptyset$.
        
        Let $D$ be a DSet that contains $B$. Then $V\setminus D\subseteq U_1$. From Lemma~\ref{lm:alg-intactness-aux} it follows through a simple inductive proof that the result $Q$ returned by the function {\sc intactNodes} satisfies $V\setminus D\subseteq Q$ or, equivalently, $V\setminus Q\subseteq D$. 
        
        We have shown above that $V\setminus Q$ is a DSet containing $B$. Hence, it is the smallest DSet that contains $B$. Thus, $Q$ is the set of all $B$-intact nodes.
    \end{proof}
    
    Let us analyze the time complexity of the function {\sc intactNodes}. In every iteration of the loop the call to {\sc quorumIntersection} dominates the time complexity of the other three calls to {\sc greatestQuorum}. In fact the former function has exponential time complexity, say, $c\cdot 2^n$ for some constant $c$, where $n$ is the number of nodes of the argument FBAS. The number $n$ for the call $\textsc{quorumIntersection}(F^{V\setminus Q})$ is $|Q|$, which is at most $|V|-i+1$ in the $i$-th iteration because $Q\subseteq U_i$ and the $U_i$ are a strictly decreasing sequence.
    
    Since there are at most $|V|+1$ iterations of the loop, we obtain that the time complexity of {\sc intactNodes} is roughly
    \begin{align*}
        \sum\nolimits_{n=0}^{|V|}c\cdot 2^n\leq c\cdot 2^{|V|+1}\;,
    \end{align*}
    twice as long as determining $\textsc{quorumIntersection}(V)$.

\section{Probabilistic intactness}\label{sec:intact-prob}
    Let $(V, S)$ be an FBAS with quorum intersection. In this section we will study the probability that a node $v\in V$ is intact. We will treat the set $B$ of ill-behaved nodes as unknown and describe it as a random variable $\mathbb B$. Let $p$ be its probability distribution, i.e., $p(B)=\mathrm P(\mathbb B=B)$. Observe that $p:\mathcal P(V)\rightarrow [0,1]$ with $\sum_{B\subseteq V}p(B) =1$. Then for every $v\in V$ we have
    \begin{align}
        \mathrm {P}(\text{$v$ is intact})&=\sum_{B\subseteq V}\,\mathrm P(\text{$v$ is intact}\mid \mathbb B = B)\cdot\mathrm P(\mathbb B = B)
        =\sum_{\substack{B\subseteq V\\ \text{$v$ is $B$-intact}}} \mathrm P(\mathbb B = B)\nonumber\\
        &=\sum_{\substack{B\subseteq V\\ \text{$\exists$ DSet $D$: $B\subseteq D\subseteq V\setminus\{v\}$}}} p(B)\;.\label{eq:def-prob-intactness}
    \end{align}
    It is particularly interesting to consider the probability of a node to be intact given that it is well-behaved. This is described by the following conditional probability:
    \begin{align}
        &\mathrm {P}(\text{$v$ is intact}\mid \text{$v$ is well-behaved})\nonumber\\
        &=\frac{\mathrm {P}(\text{$v$ is intact})}{\mathrm {P}(\text{$v$ is well-behaved})}\tag{because intact implies well-behaved}\nonumber\\
        &=\frac{\mathrm {P}(\text{$v$ is intact})}{\sum_{B\subseteq V}\mathrm P(\text{$v$ is well-behaved}\mid \mathbb B = B)\cdot \mathrm P(\mathbb B=B)}=\frac{\mathrm {P}(\text{$v$ is intact})}{\sum_{B\subseteq V\setminus\{v\}} p(B)}\;.\label{eq:def-conditional-prob-intactness}
    \end{align}
    Let us illustrate the probabilistic intactness for some examples.
    
    \begin{ex}\label{ex:prob_intact1_basis}
        Consider the symmetric simple FBAS $(V,k)$ with $V=\{a,b,c,d\}$ and $k=3$. Lemma~\ref{lem:symmDSet} yields that the DSets are $\emptyset$, $V$, and all one-element subsets of $V$. Thus,
        \begin{align}
            \mathrm P(\text{$a$ is intact}) = p(\emptyset) + p(\{b\}) + p(\{c\}) + p(\{d\})\;,\label{eq:prob_intact1_basis1}
        \end{align}
        and similarly for the nodes $b$, $c$ and $d$.
    \end{ex}
    
    \begin{ex}\label{ex:prob_intact2_basis}
        Consider the FBAS from Example~\ref{ex:FBAS_fig7}. As stated already in Example~\ref{ex:FBAS_fig7_Dsets}, the DSets are 
        $\emptyset$, $\{1,2,3\}$, $\{4,5,6\}$, $V\setminus\{7\}$ and $V$. Thus,
        \begin{align*}
            \mathrm {P}(\text{$v$ is intact})&=
            \begin{cases}
                \sum_{B\subseteq \{4,5,6\}}p(B),&\text{if $v\in\{1,2,3\}$},\\
                \sum_{B\subseteq \{1,2,3\}}p(B),&\text{if $v\in\{4,5,6\}$},\\
                \sum_{B\subseteq V\setminus\{7\}}p(B),&\text{if $v=7$}.
            \end{cases}
        \end{align*}
        In particular, $\mathrm {P}(\text{$7$ is intact}\mid \text{$7$ is well-behaved}) = 1$ by equation~\eqref{eq:def-conditional-prob-intactness}, i.e., node $7$ will not become befouled if it is well-behaved, independent from the actual probability distribution $p$. Note that this property holds for every FBAS that has a well-behaved node $v$ and a DSet that comprises all nodes except $v$. Another way to show that $\mathrm {P}(\text{$7$ is intact}\mid \text{$7$ is well-behaved})=1$ is by using the fact that $C=\{7\}$ is the greatest SCC (see Figure~\ref{fig:FBAS_fig7}): Clearly, $C$ is the set of $\emptyset$-intact nodes of $(C, S|_C)$ and then Corollary~\ref{lm:intact-SCC} implies that $7$ is intact whenever it is well-behaved.
    \end{ex}
    
    \begin{ex}\label{ex:prob_intact3_basis}
        Suppose that the FBAS is built by 4 organizations that each operate 3 nodes. We let $V=A\cup B\cup C\cup D$ with $A=\{a_1, a_2, a_3\}$ and likewise for the other organizations. We let $\mathcal O=\{A, B, C, D\}$ be the set of organizations. Similar to Example~\ref{ex:hierarchy}, we suppose that the quorum slices are built by first choosing 3 of the 4 organizations, and then choosing 2 nodes within each chosen organization. Formally, let
        \begin{align*}
            T&=\bigcup\nolimits_{\{O_1,O_2,O_3\}\in \mathcal P_3(\mathcal O)}\bigl\{M_1\cup M_2\cup M_3\mid M_i\in \mathcal P_2(O_i)\bigr\},
        \end{align*}
        with $\mathcal P_k(M)$ as in Definition~\ref{def:k-subsets}. Now we define $S(v)=\{U\in T\mid v\in U\}$ for all $v\in V$.
        The resulting set of quorums is obtained by picking 3 or 4 organizations, and then picking 2 or 3 nodes from each chosen organization. The FBAS has quorum intersection and it is easy to check that the corresponding DSets are given by $\emptyset$, $V$, each singleton set, and each set comprising a complete single organization.
        
        Let us determine $\mathrm {P}(\text{$a_1$ is intact})$. The set of all $B\subseteq V$ such that there is a DSet $D$ with $B\subseteq D\subseteq V\setminus\{a_1\}$ is given by
        \begin{align}
            \bigl\{\emptyset, \{a_2\}, \{a_3\}\bigr\}\cup\mbox{}\bigcup\nolimits_{O\in \{B,C,D\}}\mathcal P(O)\setminus\{\emptyset\}.\label{eq:example_prob_intact_groups1}
        \end{align}
        Observe that all unions in~\eqref{eq:example_prob_intact_groups1} are disjoint and hence
        \begin{align}
            \mathrm {P}(\text{$a_1$ is intact}) = p(\emptyset) + p(\{a_2\}) + p(\{a_3\}) + \sum_{O\in \{B,C,D\}}\ \sum_{\emptyset\not=O'\subseteq O}p (O').\label{eq:example_prob_intact_groups2}
        \end{align}
        Similar expressions can be derived for the other nodes in $V$ by symmetry.
    \end{ex}
    
    Let us now discuss some interesting special cases for the distribution $p$.

\subsection{At most one ill-behaved node}\label{sec:prob-intactness-one-maximal}
    If at most one node can be ill-behaved, i.e, $p(B)=0$ for $|B|>1$, we have 
    \begin{align*}
        \mathrm {P}(\text{$v$ is intact})&=p(\emptyset)+\sum\nolimits_{u\in U_v} p(\{u\})=1-\sum\nolimits_{u\in V\setminus U_v} p(\{u\}),
    \end{align*}
    where $U_v$ is the union of all DSets that do not contain $v$. Note that our assumption that $(V,S)$ has quorum intersection implies that $\emptyset$ is a DSet. Furthermore,
    \begin{align}
        \mathrm {P}(\text{$v$ is intact}\mid \text{$v$ is well-behaved})
        &=\frac{p(\emptyset)+\sum_{u\in U_v} p(\{u\})}{p(\emptyset)+\sum_{u\in V\setminus\{v\}} p(\{u\})}\nonumber\\
        &=\frac{p(\emptyset)+\sum_{u\in U_v} p(\{u\})}{1-p(\{v\})}\;.\label{eq:cond-prob-intactness-one-maximal}
    \end{align}
    
    \begin{ex}[continuation of Example~\ref{ex:prob_intact1_basis}]\label{ex:prob_intact1}
        Suppose that
        \begin{align*}
            p(\emptyset)&=0.6,&
            p(\{a\})&=0.2,&
            p(\{b\})&=p(\{c\})=0.1,&
            p(\{d\})&=0.
        \end{align*}
        By~\eqref{eq:prob_intact1_basis1}, $\mathrm {P}(\text{$a$ is intact})=0.8$. Likewise, $\mathrm {P}(\text{$b$ is intact})=\mathrm {P}(\text{$c$ is intact})=0.9$ and $\mathrm {P}(\text{$d$ is intact})=1$. With~\eqref{eq:cond-prob-intactness-one-maximal} we conclude furthermore that the conditional probability of each node to be intact given that it is well-behaved is $1$ because $U_v=V\setminus\{v\}$ for each node $v$.
        
        Note that the latter property carries over to a large class of symmetric simple FBAS $(V,k)$: If $|V|-1\geq k\geq |V|/2 +1$, then Lemma~\ref{lem:symmDSet} implies that every one-element subset of $V$ is a DSet. We thus get $U_v=V\setminus \{v\}$ for every $v\in V$ and consequently $\mathrm {P}(\text{$v$ is intact}\mid \text{$v$ is well-behaved})=1$.
    \end{ex}
    
    \begin{ex}[continuation of Example~\ref{ex:prob_intact2_basis}]\label{ex:prob_intact2}
        Since the DSets are $\emptyset$, $\{1,2,3\}$, $\{4,5,6\}$, $V\setminus\{7\}$, and $V$, we obtain
        \begin{align*}
            U_1&=U_2=U_3=\{4,5,6\},&U_4&=U_5=U_6=\{1,2,3\},&U_7=V\setminus\{7\}.
        \end{align*}
        From this we conclude that
        \begin{align*}
            \mathrm {P}(\text{$v$ is intact})=
            \begin{cases}
                p(\emptyset)+p(\{4\})+p(\{5\})+p(\{6\})\;,&\text{if $v\in\{1,2,3\}$},\\
                p(\emptyset)+p(\{1\})+p(\{2\})+p(\{3\})\;,&\text{if $v\in\{4,5,6\}$},\\
                1-p(\{7\})\;,&\text{if $v=7$}.
            \end{cases}
        \end{align*}
    \end{ex}

\subsection{Independent failures}\label{sec:prob-intact-indep}
    Let us suppose that ill-behavior only occurs due to random and non-coordinated failures of nodes. We can model this case by assuming that the events of distinct nodes to be ill-behaved are independent, and that a single node $v$ becomes ill-behaved with probability $p_v$. Thus, for every $B\subseteq V$ we have 
    \begin{align*}
        p(B)=\prod_{u\in B}p_u\cdot\prod_{u\in V\setminus B}(1-p_u)\;,
    \end{align*}
    and therefore
    \begin{align*}
        \mathrm {P}(\text{$v$ is intact})&=\sum_{\substack{B\subseteq V\\ \text{$\exists$ DSet $D$: $B\subseteq D\subseteq V\setminus\{v\}$}}} \prod_{u\in B}p_u\cdot\prod_{u\in V\setminus B}(1-p_u)\;.
    \end{align*}
    Note that $$\sum_{B\subseteq V\setminus\{v\}}p(B)=\sum_{B\subseteq V\setminus\{v\}}\prod_{u\in B}p_u\cdot\prod_{u\in V\setminus B}(1-p_u)=1-p_v,$$ 
    which is easy to show by induction. Then~\eqref{eq:def-conditional-prob-intactness} yields
    \begin{align*}
        \mathrm {P}(\text{$v$ is intact}\mid \text{$v$ is well-behaved})&=\frac{\mathrm {P}(\text{$v$ is intact})}{1-p_v}\;.
    \end{align*}
    
    \begin{ex}[continuation of Example~\ref{ex:prob_intact1_basis}]
       Suppose that 
        \begin{align*}
            p_a&=0.2,&
            p_b&=p_c=0.1,&
            p_d&=0,
        \end{align*}
        then by~\eqref{eq:prob_intact1_basis1},
        \begin{align*}
            \mathrm {P}(\text{$a$ is intact})=0.8\cdot 0.9^2 + 0.8\cdot 0.1\cdot 0.9 + 0.8\cdot 0.9\cdot 0.1 + 0 = 0.792.
        \end{align*}
        Similarly,
        \begin{align*}
            \mathrm {P}(\text{$b$ is intact})&=\mathrm {P}(\text{$c$ is intact})=0.882,&\mathrm {P}(\text{$d$ is intact})&=0.954.
        \end{align*}
        For the conditional probabilities we have
        \begin{align*}
            \mathrm {P}(\text{$a$ is intact}\mid \text{$a$ is well-behaved})&=0.99,\\
            \mathrm {P}(\text{$b$ is intact}\mid \text{$b$ is well-behaved})&=0.98,\qquad\mbox{(same for $c$)}\\
            \mathrm {P}(\text{$d$ is intact}\mid \text{$d$ is well-behaved})&=0.954.
        \end{align*}
        Observe that the conditional probability for $a$ to be intact is higher than for the other three nodes. This is because we assume that $a$ is well-behaved and therefore it can only become befouled if at least one of the other nodes is ill-behaved. However, these nodes are more likely to be well-behaved than node~$a$.
    \end{ex}
    
    \begin{ex}[continuation of Example~\ref{ex:prob_intact2_basis}]
        For node $1$ there are precisely two DSets $D$ with $D\subseteq V\setminus\{1\}$, namely the sets $\emptyset$ and $\{4,5,6\}$. Then
        \begin{align*}
            \mathrm {P}(\text{$1$ is intact})&=\textstyle{\sum_{B\subseteq\{4,5,6\}}\prod_{u\in B}p_u\cdot\prod_{u\in V\setminus B}1-p_u}\\&=(1-p_1)\cdot(1-p_2)\cdot(1-p_3)\cdot(1-p_7),
        \end{align*}
        and
        \begin{align*}
            \mathrm {P}(\text{$1$ is intact}\mid \text{$1$ is well-behaved})&=(1-p_2)\cdot(1-p_3)\cdot(1-p_7).
        \end{align*}
        The computation is similar for the other nodes in $V\setminus\{7\}$.
    \end{ex}

\subsection{Grouping of nodes}\label{sec:prob-intact-groups}
    Let us generalize the previous case by assuming that the nodes are associated with entities or organizations (as in Example~\ref{ex:hierarchy} or Example~\ref{ex:prob_intact3_basis}) that act independently. Therefore, we can assume that the ill-behavior of different organizations is statistically independent. More precisely, we assume a partitioning $\mathcal O$ of $V$, where every element of the partitioning consists of the nodes of a single organization, and a probability distribution $p_O$ over $\mathcal P(O)$ for every $O\in \mathcal O$. Then we define
    \begin{align*}
        p(B)=\prod_{O\in \mathcal O} p_O(B\cap O).
    \end{align*}
    
    In practice the quorum slices of nodes within one organization will presumably be set up in a way that all nodes of the organization trust each other; equivalently, for every organization there is an SCC of the trust graph containing all nodes of the organization. This implies that every trust cluster of the FBAS is a union of organizations due to Lemma~\ref{lem:trust-cluster}. In this case we can simplify the computation of the probabilistic intactness of a node within a trust cluster -- it is determined completely by the nodes of the trust cluster and independent from the nodes outside the trust cluster.

    \begin{lem}\label{lem:prob-intact-groups}
        Suppose that for every organization $O\in \mathcal O$ there is an SCC that contains $O$. Let $Z$ be a trust cluster and $v\in Z$. Then
        \begin{align*}
            \mathrm P(\text{$v$ is intact})=\sum_{\substack{B\subseteq Z\\ \text{$\exists$ DSet $D$ of $(Z, S|_Z)$: $B\subseteq D\subseteq Z\setminus\{v\}$}}} \prod_{O\in \mathcal O, O\subseteq Z}p_O(B\cap O)\;.
        \end{align*}
    \end{lem}
    \begin{proof}
        First observe that Lemma~\ref{lem:trust-cluster} implies that for every organization $O\in \mathcal O$ either $O\subseteq Z$ or $O\subseteq V\setminus Z$. Let us distinguish between these two kinds of organizations and put $\mathcal O_\mathrm {in}=\{O\in \mathcal O\mid O\subseteq Z\}$ and $\mathcal O_\mathrm {out}=\{O\in \mathcal O\mid O\subseteq V\setminus Z\}$. Then $\mathcal O=\mathcal O_\mathrm{in}\cup \mathcal O_\mathrm{out}$. We derive
        \begin{align*}
            &\mathrm {P}(\text{$v$ is intact})=\sum_{\substack{B\subseteq V\\ \text{$v$ is $B$-intact in $(V,S)$}}} p(B)\\
            &=\sum_{\substack{B\subseteq V\\ \text{$v$ is $B\cap Z$-intact in $(Z,S|_Z)$}}} p(B)\tag{by Theorem~\ref{theo:intact-SCC}}\\
            &=\sum_{\substack{B_\mathrm{in}\subseteq Z\\ \text{$v$ is $B_\mathrm{in}$-intact in $(Z,S|_Z)$}}} \sum_{B_\mathrm{out}\subseteq V\setminus Z}p(B_\mathrm{in}\cup B_\mathrm{out})\\
            &=\sum_{\substack{B_\mathrm{in}\subseteq Z\\ \text{$v$ is $B_\mathrm{in}$-intact in $(Z,S|_Z)$}}} \sum_{B_\mathrm{out}\subseteq V\setminus Z}\ \ \prod_{O\in \mathcal O} p_O((B_\mathrm{in}\cup B_\mathrm{out})\cap O)\\
            &=\sum_{\substack{B_\mathrm{in}\subseteq Z\\ \text{$v$ is $B_\mathrm{in}$-intact in $(Z,S|_Z)$}}} \sum_{B_\mathrm{out}\subseteq V\setminus Z}\Bigl(\prod_{O\in \mathcal O_\mathrm{in}} p_O(B_\mathrm{in}\cap O)\Bigr)\cdot\Bigl(\prod_{O\in \mathcal O_\mathrm{out}} p_O(B_\mathrm{out}\cap O)\Bigr)\\
            &=\sum_{\substack{B_\mathrm{in}\subseteq Z\\ \text{$v$ is $B_\mathrm{in}$-intact in $(Z,S|_Z)$}}}\Bigl(\prod_{O\in \mathcal O_\mathrm{in}} p_O(B_\mathrm{in}\cap O)\Bigl)\cdot\sum_{B_\mathrm{out}\subseteq V\setminus Z}\Bigl(\prod_{O\in \mathcal O_\mathrm{out}} p_O(B_\mathrm{out}\cap O)\Bigr)
        \end{align*}
        It suffices to show that $\sum_{B_\mathrm{out}\subseteq V\setminus Z}\prod_{O\in \mathcal O_\mathrm{out}} p_O(B_\mathrm{out}\cap O)=1$. Let $O_\mathrm{out}=\{O_1, \ldots, O_n\}$ and observe that $\bigcup_{i=1}^n O_i=V\setminus Z$. Then 
        \begin{align*}
            \sum\nolimits_{B_\mathrm{out}\subseteq V\setminus Z}\prod\nolimits_{O\in \mathcal O_\mathrm{out}} p_O(B_\mathrm{out}\cap O)&=\sum\nolimits_{B_1\subseteq O_1, \ldots, B_n\subseteq O_n}\prod\nolimits_{i=1}^n p_{O_i}(B_i)\\
            &=\prod\nolimits_{i=1}^n \sum\nolimits_{B_i\subseteq O_i} p_{O_i}(B_i)=1,
        \end{align*}
        which concludes the proof.
    \end{proof}
    Note that Lemma~\ref{lem:prob-intact-groups} also applies to the probability distributions discussed in Section~\ref{sec:prob-intact-indep} as they are special cases of the probability distributions in this section.
    
    \begin{ex}[continuation of Example~\ref{ex:prob_intact3_basis}]\label{ex:prob_intact3}
        Let us study our example for the partitioning $\mathcal O$. For every $O\in\mathcal O$ let $e_O=p_O(\emptyset)$. Then we can reformulate the right hand side of~\eqref{eq:example_prob_intact_groups2} as
        \begin{align*}
            &\bigl(e_A + p_A(\{a_2\}) + p_A(\{a_3\})\bigr)e_B e_C e_D
             + \sum_{O\in \{B,C,D\}}\ \sum_{\emptyset\not=O'\subseteq O}p_O(O')\prod_{O''\in \mathcal O\setminus \{O\}}e_{O''}\\
             &=\bigl(e_A + p_A(\{a_2\}) + p_A(\{a_3\})\bigr)e_B e_C e_D
             + \sum_{O\in \{B,C,D\}}\ (1-e_O)\prod_{O''\in \mathcal O\setminus \{O\}}e_{O''}\tag{$\star$}\\
             &=\bigl(e_A + p_A(\{a_2\}) + p_A(\{a_3\})\bigr)e_B e_C e_D\\
             &\phantom{=\ }\mbox{}+ (1-e_B)e_A e_C e_D+(1-e_C)e_A e_B e_D + (1-e_D)e_A e_B e_C.
        \end{align*}
        At ($\star$) we used that fact that $p_O$ is a probability distribution over $\mathcal P(O)$ and therefore $\sum_{\emptyset\not=O'\subseteq O}p_O(O') = 1-e_O$.
        
        Suppose that $e_A=e_B=e_C=e_D$ and let us refer to this value as $e$. Then
         \begin{align}
            \mathrm P(\text{$a_1$ is intact})=e^3\bigl(p_A(\{a_2\}) + p_A(\{a_3\})+3-2e\bigr).\label{eq:prob_intact3}
        \end{align}
    \end{ex}

\subsection{Grouping of nodes with Byzantine failures}
    Finally, let us consider a specialization of the previous case and assume a simple model of the probability distribution $p_O$ of every element $O\in \mathcal O$. Our model is based on the assumption that nodes within an organization can become ill-behaved for two reasons: (i)~random failures that are statistically independent and (ii)~the organization experiences a Byzantine fault and all nodes of the organization become ill-behaved.
    
    For every $O\in \mathcal O$ we presuppose the probability $q_O$ (the probability of a single node of the organization to fail randomly) and the probability $r_O$ (the probability of all the nodes of the organization $O$ to become Byzantine). Then for every $B\subseteq O$ we define
    \begin{align*}
        p_O(B)=\begin{cases}r_O + (1-r_O)\cdot (q_O)^{|O|}\;,&\text{if $B=O$,}\\(1-r_O)\cdot (q_O)^{|B|}\cdot (1-q_O)^{|O|-|B|}\;,&\text{otherwise.}\end{cases}
    \end{align*}
    Observe that $p_O$ is indeed a probability distribution.
    
    \begin{ex}[continuation of Example~\ref{ex:prob_intact3}]\label{ex:prob_intact3_2}
        Let $q_O$ and $r_O$ be the probabilities for organization $O$. Let us suppose that all $q_O$ are identical for each organization and likewise for $r_O$, and let us simply refer to these values as $q$ and $r$. Then $e=p_O(\emptyset)=(1-r)\cdot (1-q)^3$. Furthermore, $p_A(\{a_2\})=p_A(\{a_3\})=(1-r)\cdot q\cdot (1-q)^2$. Then from~\eqref{eq:prob_intact3} we derive
        \begin{align*}
            \mathrm P(\text{$a_1$ is intact})=(1-r)^3 (1-q)^9\bigl(2(1-r) q (1-q)^2+3-2(1-r) (1-q)^3\bigr)\;.
        \end{align*}
        For example, if $q=0.1$ and $r=0.01$, then $\mathrm P(\text{$a_1$ is intact})\approx 0.65$. By symmetry all nodes have the same probability to be intact.
    \end{ex}
    
    The FBAS in Example~\ref{ex:prob_intact3_2} is not the only way to arrange four organizations with three nodes per organization. We could also lay out the nodes as a symmetric simple FBAS. We will study this FBAS in the following example.
    
    \begin{ex}\label{ex:prob_intact4}
        Consider the symmetric simple FBAS $(V, 8)$, where $V=A\cup B\cup C\cup D$ is defined as in Example~\ref{ex:prob_intact3_basis}. From Lemma~\ref{lem:symmDSet} we know that the DSets are $V$ and every subset of $V$ having at most 3 elements. Therefore, for every $v\in V$,
        \begin{align*}
            \mathrm P(\text{$v$ is intact}) = \sum_{B\subseteq V\setminus\{v\}, |B|\leq 3}p(B).
        \end{align*}
        This sum is combinatorially more complex than the expression in Example~\ref{ex:prob_intact3_2}. However, assuming identical probability values $q$ and $r$ for each organization as in Example~\ref{ex:prob_intact3_2}, we can group the sets $B$ into categories as many of them have the same value for $p(B)$. We consider the following five categories of sets $B$: (i)~having $0$ members, (ii)~having $1$ member, (iii)~having $2$ members, (iv)~having $3$ members, not all in one organization, and (v)~a complete organization. Then we find that there is 1 set in category~(i), 11 sets in category~(ii), 55 sets in category~(iii), 162 sets in category~(iv), and 3 sets in category~(v). For example, all sets in category~(iii) satisfy $p(B)=(1-r)^4\cdot q^2\cdot(1-q)^{10}$, so that we obtain
        \begin{align*}
            \mathrm P(\text{$v$ is intact}) &= (1-r)^4\cdot(1-q)^{12} + 11\cdot (1-r)^4\cdot q\cdot(1-q)^{11} \\
            &\phantom{=\ }\mbox{}+ 55\cdot (1-r)^4\cdot q^2\cdot(1-q)^{10} + 162\cdot (1-r)^4\cdot q^3\cdot(1-q)^{9}\\
            &\phantom{=\ }\mbox{}+3\cdot (1-r)^3\cdot(1-q)^{9}\cdot\bigl(r + (1-r)\cdot q^3\bigr).
        \end{align*}
        For the values $q=0.1$ and $r=0.01$ as in Example~\ref{ex:prob_intact3_2} we get $\mathrm P(\text{$v$ is intact})\approx 0.86$ for every $v\in V$. We observe that the probability that a node becomes befouled in case of the symmetric simple FBAS is \emph{significantly smaller} (by a factor of approximately 2.5) than for the hierarchical FBAS in Example~\ref{ex:prob_intact3_basis}. 
    \end{ex}

\section{Concluding remarks}

In this paper we have discussed the main concepts of the federated Byzantine agreement system as defined by Mazi\`eres~\cite{Maz16}, in particular quorum and quorum intersection, and intactness of nodes. We have illustrated these concepts with several examples, ranging from small ``academic'' cases to the greatest strongly connected component of the current Stellar network; see Example~\ref{ex:stellar-core}. In our discussion of intactness, we have clarified the relation between Mazi\`eres' original definition and the more recent definitions in~\cite{LosGafMaz19,LokLosMaz19}, and we have introduced the \emph{subslice property} as a condition for equivalence of the two concepts; see Section~\ref{sec:intact-compare}.

We have treated the problems of quorum enumeration and quorum intersection as well as the intactness of nodes from an algorithmic point of view. We have shown that the quorum split decision problem is NP-complete even for the special class of \emph{simple FBAS}, which we introduced in this paper; see Definition~\ref{def:simple}. Based on the work of Lachowski~\cite{Lac19}, we have derived algorithms for quorum enumeration, checking quorum intersection, and computing intact nodes. Some computed examples with our implementations of these algorithms in the Python package \emph{Stellar Observatory}\/\footnote{\url{https://github.com/andrenarchy/stellar-observatory}} are given is this paper; see, e.g., Example~\ref{ex:stellar-core} and Figure~\ref{fig:quorum-enumeration-computation}. In order to derive the algorithms, prove their correctness and reduce their complexity we have introduced several new concepts including the \emph{trust graph} and \emph{trust clusters} of an FBAS; see Section~\ref{sec:trust-graph}. 

In the final section of this paper we have treated intactness from a probabilistic point of view, and we have studied the probability that a node is intact in different scenarios of ill-behaved nodes. Of particular interest in this context is the intactness probability in hierarchical FBAS settings as the ones in Example~\ref{ex:hierarchy}, which we then discussed in Examples~\ref{ex:prob_intact3_basis},~\ref{ex:prob_intact3}, and~\ref{ex:prob_intact3_2}. 

In Example~\ref{ex:prob_intact4} we have shown that from the viewpoint of the intactness probability, the symmetric simple FBAS is superior to a hierarchical FBAS. This observation not only holds for the simplified probability distribution we considered, but in fact for every probability distribution $p$, which is immediately clear from the definition of probabilistic intactness in~\eqref{eq:def-prob-intactness}: The set of DSets in Example~\ref{ex:prob_intact3_basis} is a proper subset of the set of DSets in Example~\ref{ex:prob_intact4}. Therefore, no probabilistic intactness in the hierarchical FBAS exceeds the corresponding probabilistic intactness in the symmetric simple FBAS. 

In Sections~\ref{sec:quorum-intersection-scc} and~\ref{sec:algo-intact-scc} we have shown that intersection and intactness properties of the greatest SCC (and more generally any trust cluster) of an FBAS are independent from the remainder of the FBAS. Therefore, this component of the FBAS can be studied in isolation, and this property carries over to probabilistic intactness; see~Lemma~\ref{lem:prob-intact-groups}. A similar argument as in Example~\ref{ex:prob_intact4} then shows that 
the probabilistic intactness of nodes in the hierarchical setup of the greatest SCC of the FBAS in Example~\ref{ex:stellar-core}
(with quorum slices containing 11, 12 or 13 nodes) will not be larger than in a symmetric simple FBAS with $k=12$.



\section*{Acknowledgement}
This work was supported by a grant in the SDF Academic Research Program.

\printbibliography

\end{document}